%% file: main.tex
\pdfoutput=1
\documentclass[aps,prl,reprint,superscriptaddress,nofootinbib,floatfix,longbibliography]{revtex4-2}
\usepackage[T1]{fontenc}
\usepackage{lmodern}
\usepackage{amsmath,amssymb,mathtools,bm,mathrsfs,graphicx,microtype,xcolor}
\usepackage{booktabs,tabularx,array}
\usepackage[colorlinks=true,allcolors=blue!55!black]{hyperref}
\hypersetup{pdftitle={Preparation Changes the Cost of Calibration for Quantum Control},pdfauthor={Xiu-Hao Deng}}

\input{macros}

\newcolumntype{Y}[1]{>{\raggedright\arraybackslash}p{#1}}
\begin{document}
\title{Preparation Changes the Cost of Calibration for Quantum Control}
\author{Xiu-Hao Deng}
\email{dengxiuhao@iqasz.cn}
\thanks{ORCID: \href{https://orcid.org/0000-0003-0105-9074}{0000-0003-0105-9074}}
\affiliation{International Quantum Academy, Shenzhen 518048, China}
\affiliation{Shenzhen Branch, Hefei National Laboratory, Shenzhen 518048, China}
\date{September 19, 2026}
\begin{abstract}
Error-correction records can reveal noise yet supply the information needed
for control too slowly. In a surface code under Gaussian dephasing, encoded
calibration learns a correlation that reverses the preferred control mainly
through rare events. A product probe exposes it more often through the same
checks, reducing ideal-record sample complexity from inverse-square to
inverse-linear in phase variance. Counting both calibration and execution
costs, we show that this gain can repay probe overhead and enable protected
tasks within budgets that exclude encoded calibration.
\end{abstract}
\maketitle

Quantum error correction generates syndrome and detection records while
protecting quantum information. Reusing them avoids a separate probe
preparation and supports learning during error
correction \cite{combes2014insitu,kelly2016scalable}. Detection events also
support ongoing reinforcement-learning control of error correction
\cite{sivak2026rlqec}. These records can reveal physical
and logical noise
\cite{wagner2021optimal,wagner2022pauli,wagner2023logical,takou2025detector,zheng2026efficient}.
Dedicated probes cost reset, preparation, and acquisition time, but can be
designed to expose control-relevant correlations faster. Both costs matter
when calibration and the protected operation share a finite budget.

Before an unknown logical input, possibly reference-entangled, arrives,
calibration commits a control sequence [Fig.~\ref{fig:mechanism}(a)]. The
protected logical idle must then last $L$ error-correction cycles within
total budget $B$ and a service deadline. At every prefix, tolerance $D$ bounds
$\frac12\lVert\overline{\mathcal C}-\mathcal I\rVert_\diamond$ for the
calibration-averaged channel, worst case over input and reference
\cite{watrous2018theory}.
Service syndromes drive the fixed recovery without revising the sequence.
Both routes are charged for calibration acquisition. The aim is a useful
control decision, without full noise reconstruction.

The comparison fixes the encoding, cycle duration, extractor, recovery, and
finite classical action menu. Encoded queries may use any logical/reference
input and return every check plus the recovered output; adaptive processing
and external quantum memory are allowed. Product calibration instead changes
the data input preparation and uses the same extractor. Alternative physical
check circuits and query durations are outside this comparison.

Correctability alone does not cause slow learning: frequent syndrome rates
generally reveal first-order covariance projections. Our construction places
the control-relevant correlation outside those projections. The two models
then favor opposite controls but encoded inputs distinguish them only through
rare two-error sectors [Fig.~\ref{fig:mechanism}(b)]. A product input, carrying
no unknown logical state, exposes the same correlation more frequently
[Fig.~\ref{fig:mechanism}(c)]. Both use the same checks and retain every
outcome; the encoded comparator also receives the recovered quantum output
and may use free external quantum memory. Changing preparation can save
enough calibration cycles to repay its overhead and alter task feasibility.

\begin{figure*}[t]
\includegraphics[width=\textwidth]{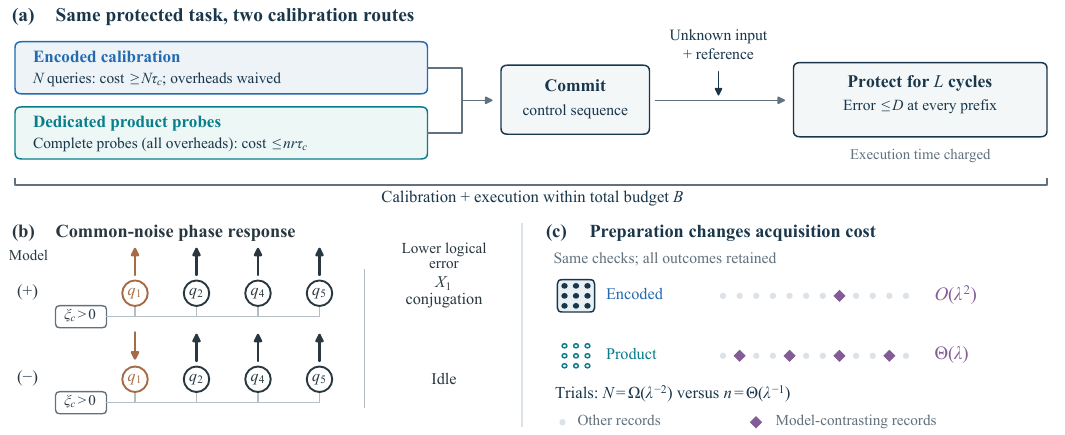}
\caption{Preparation changes the cost of learning a useful control.
(a) The two branches are alternative calibration routes, completed before the unknown logical/reference input, that commit a control for the same protected idle. $N$ and $n$ count encoded queries and product probes; $\tau_c$ is one protected-cycle duration, and $r\tau_c$ upper-bounds a complete product probe. The budget charges calibration and execution; encoded overheads are waived in the lower bound.
(b) Arrows show phase responses to a positive common fluctuation $\xi_c$; independent fluctuations are omitted. Only qubit 1 (copper) reverses between models. All arrows reverse when $\xi_c$ changes sign. The right column gives the lower-error control (ideal pulses).
(c) Boxed filled and unboxed open nine-dot icons denote encoded and product preparations of the same data qubits. Each strip symbol represents one trial's check record. Purple diamonds highlight rare, model-sensitive outcomes; gray circles denote other records. Both are retained. Individual highlighted outcome categories have finite relative model contrast; with ideal delivered records, their probabilities scale as $O(\lambda^2)$ for encoded and $\Theta(\lambda)$ for product preparation. Product discrimination uses the joint check values within these events, not their total count alone. $N$ and $n$ count calibration trials, not equal elapsed time; the displayed orders concern fixed-error model discrimination. Both routes return all eight checks; encoded queries also return the recovered output. Strips are schematic, not sample data or equal-cost records.}
\label{fig:mechanism}
\end{figure*}

Designed syndrome-extraction experiments and multiqubit spectroscopy already
learn circuit noise and cross correlations with product-state probes
\cite{hockings2025scalable,pazsilva2017multiqubit,szankowski2016spectroscopy}.
Our contribution is the acquisition cost of a control decision and the
full-output converse, rather than the correlation probe itself. This
complements efficient Pauli-noise characterization
\cite{flammia2020efficient,harper2020efficient} and learning advantages from
quantum processing or memory \cite{huang2022advantage,chen2021memory}.
Coherent paths can interfere within a syndrome sector
\cite{huang2019performance}; Pauli conjugation can change that interference
and the logical error \cite{cai2020pauli}. This supplies the decision-relevant
correlation in the construction below.

\emph{A correlation that changes the control.---}
Consider a rotated nine-data-qubit surface code with eight reusable check
ancillas \cite{tomita2014surface}. Qubits are indexed $0,\ldots,8$. During one
idle interval, centered Gaussian phases have covariance
\begin{equation}
 K_\theta=\lambda(I+v_\theta v_\theta^{\mathsf T}),\qquad
 v_\theta=(0,\theta,1,0,1,1,0,0,0)^{\mathsf T},
 \label{eq:noise}
\end{equation}
where $\theta\in\{+1,-1\}$ is stable through calibration and service, and
$\lambda$ sets the phase variance. Each phase combines independent and signed
common Gaussian fluctuations. Local-plus-common models are standard descriptions of
correlated dephasing \cite{yoshihara2010correlated,vonlupke2020spectroscopy};
Eq.~\eqref{eq:noise} defines a controlled theoretical instance.

The actions are idle ($a=0$) and $X_1$ conjugation of the idle ($a=1$), with
the known frame restored before comparison and the same fixed recovery. Let
$K_\theta=\lambda\Sigma_{\theta,0}$ denote the idle covariance. Ideal
conjugation gives $\Sigma_{\theta,1}=D_1\Sigma_{\theta,0}D_1$, where $D_1$
reverses coordinate 1. Let
$q_{\theta,a}$ be the recovered logical phase-flip probability for model
$\theta$ and action $a$. The ideal-pulse physical response is
\begin{equation}
\begin{split}
 q_{\theta,0}-q_{\theta,1}
 &=\frac14\,\mathbb E_\theta[\phi_1\phi_2\phi_4(\phi_5+\phi_8)]
   +O(\lambda^3)\\
 &=\frac{3\theta}{4}\lambda^2+O(\lambda^3).
\end{split}
\label{eq:physical-response}
\end{equation}
The two weight-four $Z$ stabilizers containing site 1 have supports
$\{1,2,4,5\}$ and $\{1,2,4,8\}$. For error supports $E,F$, interference
within one syndrome and recovered logical class requires their symmetric
difference to be a $Z$ stabilizer; conjugation changes its sign only if it
contains site 1. No weight-two $Z$ stabilizer contains this site,
so the first allowed response is fourth order in phase, or second order in
$\lambda$. The Supplemental Material (SM) derives this selection rule for
one-logical-qubit Calderbank--Shor--Steane (CSS) codes under the stated parity
and recovery assumptions \cite{prlv4supp}. The allowed terms can vanish or cancel; here
they yield Eq.~\eqref{eq:physical-response}. The leading logical dephasing
rates are
$q_{+,0}=45\lambda^2/16$, $q_{-,0}=33\lambda^2/16$,
$q_{+,1}=q_{-,0}$, and $q_{-,1}=q_{+,0}$, up to $O(\lambda^3)$.
Thus conjugation helps model $+$ and hurts model $-$. Gaussian fourth moments
depend on the full covariance through Wick's rule, while the leading syndrome
rates access only the projections in Eq.~\eqref{eq:leading}. The correlation
affecting the control is therefore absent from the leading syndrome rates,
although it is present in higher-order syndrome statistics.

Comparing encoded idle and conjugated cycles---an A/B experiment---can
measure the leading control response. With
$\Delta P_\theta(s)=P_{\theta,0}(s)-P_{\theta,1}(s)$ denoting the idle-minus-
conjugated probability difference, this patch obeys
\begin{equation}
 q_{\theta,0}-q_{\theta,1}
 =\tfrac12[\Delta P_\theta(10)+\Delta P_\theta(11)]+O(\lambda^3).
 \label{eq:syndrome-response}
\end{equation}
Here 10 and 11 denote outcomes $(+,-,+,-)$ and $(-,-,+,-)$ for the ordered
checks $X_0X_1X_3X_4$, $X_4X_5X_7X_8$, $X_6X_7$, and $X_1X_2$; a minus sign
is bit 1, least-significant bit first. These two syndromes have probability
$\Theta(\lambda^2)$ and a finite relative contrast, so a fixed-error decision
from their counts needs $\Theta(\lambda^{-2})$ independent cycles. The
obstacle is slow acquisition, not absence of an observable response. This
relation uses ideal pulses and the fixed recovery, as does
Eq.~\eqref{eq:physical-response}.

\emph{When syndrome learning is slow.---}
Let a fixed stabilizer code with encoding $V$ on $m$ physical qubits exactly
correct the coherent span of
$I,Z_0,\ldots,Z_{m-1}$ \cite{knill1997theory}. Each cycle applies
\begin{equation}
 U(\phi)=\exp\!\left[-\frac{i}{2}\sum_{j=0}^{m-1}\phi_jZ_j\right],
 \qquad \phi\sim\mathcal N(0,\lambda\Sigma_{\theta,a}),
\end{equation}
followed by ideal complete syndrome extraction and a fixed correcting
recovery. The covariance shapes are known, positive definite, and selected
from a fixed compact family; the classically selected action menu is finite,
and cycles are independent conditional on the stable label $\theta$.

For a nonzero single-$Z$ syndrome $s$, let $J_s$ collect its qubits.
Correctability gives
$Z_jV=\nu_{s,j}Z_{j_s}V$ for $j\in J_s$, where $\nu_{s,j}=\pm1$.
With $\nu_s$ zero outside $J_s$, its leading-probability coefficient is
\begin{equation}
\omega_{\theta,a,s}=\tfrac14\nu_s^{\mathsf T}
                   \Sigma_{\theta,a}\nu_s .
\label{eq:leading}
\end{equation}
Correctability determines these projections but does not make them equal
between models. Matching is an additional restriction: it can leave a
covariance direction invisible to frequent syndromes even when fourth-moment
path interference makes that direction relevant to the preferred control.

We use Kullback--Leibler (KL) divergence with natural logarithms.

\begin{theorem}[Syndrome-learning cost]
\label{thm:matching}
If $\omega_{+,a,s}=\omega_{-,a,s}$ for every admitted action and every
nonzero single-error syndrome, any syndrome-only adaptive experiment with at
most $N$ cycles has record $Y_N$ satisfying, in both directions,
\begin{equation}
 D_{\rm KL}(P_+^{Y_N}\Vert P_-^{Y_N})\le K_{\rm syn}N\lambda^2 .
 \label{eq:syndrome-bound}
\end{equation}
Here $K_{\rm syn}$ is independent of $N$ and sufficiently small $\lambda$.
Logical states may persist and undergo syndrome-dependent trace-preserving
logical operations. Identifying the label with error at most
$0<\epsilon<1/2$ in each model therefore needs at least
${\rm kl}(1-\epsilon\Vert\epsilon)/(K_{\rm syn}\lambda^2)$ cycles,
where ${\rm kl}$ denotes binary relative entropy.
\end{theorem}

Matching removes the leading model contrast from the frequent single-error
syndromes; sectors first reached by two errors can then supply the leading
information. Uniform conditional-state bounds and the adaptive chain rule
prove Eq.~\eqref{eq:syndrome-bound} \cite{prlv4supp}. For
$\delta=\max_{a,s}|\omega_{+,a,s}-\omega_{-,a,s}|$, uniform positive
constants $C_1,C_2$ give, for sufficiently small $\lambda$,
\begin{equation}
 D_{\rm KL}(P_+^{Y_N}\Vert P_-^{Y_N})
 \le N(C_1\lambda\delta^2+C_2\lambda^2),
 \label{eq:approx-matching}
\end{equation}
and likewise in reverse, uniformly over the declared compact covariance
family. A mismatch
$\delta=O(\sqrt\lambda)$ preserves the inverse-square lower bound, whereas a
fixed nonzero mismatch restores first-order information at sufficiently weak
noise.

Equal individual variances do not by themselves imply matching: degenerate
single-qubit errors sharing a syndrome also reveal off-diagonal covariance
combinations through Eq.~\eqref{eq:leading}. For this patch, fixed covariance
diagonals and only the ideal idle/$X_1$ menu leave 36 off-diagonal coordinates
and two independent leading-rate constraints. The resulting 34-dimensional
kernel describes the ambiguity remaining after learning them. Other interfaces
select their own covariance projections; the kernel dimension does not assert
that every invisible direction changes the useful control.

The matching cost and the leading control benefit can vary independently. In
the present patch, adding
$\theta t(e_0e_3^{\mathsf T}+e_3e_0^{\mathsf T})$ to the covariance shape,
with $e_j$ the $j$th coordinate vector, produces leading-rate mismatch
$\delta=|t|$ while preserving
the ideal leading control gap $3\lambda^2/4$ in both models and leaving the
product-probe law unchanged throughout the uniformly positive covariance
family $|t|\le1/4$ \cite{prlv4supp}. For ideal idle syndrome records,
$k_{\rm syn}=A_0\lambda^2+\lambda t^2+O(\lambda^{5/2})$ along
$t=c_t\sqrt\lambda$ with fixed $c_t$, where
$A_0=(5/2)\log(5/2)-3/2$. At $\lambda=10^{-3}$, the two leading
terms become comparable near $|t|=\sqrt{A_0\lambda}\simeq0.028$.
This asymptotic crossover and covariance range do not specify a finite-budget
tolerance. The full-output cost certificate
below uses $t=0$. At exact matching,
the nonzero coefficient $A_0$ and the rare-syndrome test above show that the
inverse-square learning cost is attained in this example.

A product probe exposes a different projection. Prepare $|++\rangle$ on a
pair with common individual variances and unequal covariance
$b_\theta=\Sigma_{\theta,ij}$, then jointly read the commuting parities
$X_iX_j$ and $Z_iZ_j$, denoted $x,z$. With
$a_p=(\Sigma_{ii}+\Sigma_{jj})/2>|b_\theta|$,
\begin{equation}
 P_\theta(x=-1,z)=\tfrac{\lambda}{4}(a_p+zb_\theta)+O(\lambda^2).
 \label{eq:rare}
\end{equation}
The conditional $z$ laws differ by a finite amount in a sector of probability
$\Theta(\lambda)$. Both KL divergence and squared Hellinger distance are
therefore $\Theta(\lambda)$, giving fixed-error testing with
$\Theta(\lambda^{-1})$ independent probes. This may require entangling
readout despite the product input.

\begin{figure*}[t]
\includegraphics[width=\textwidth]{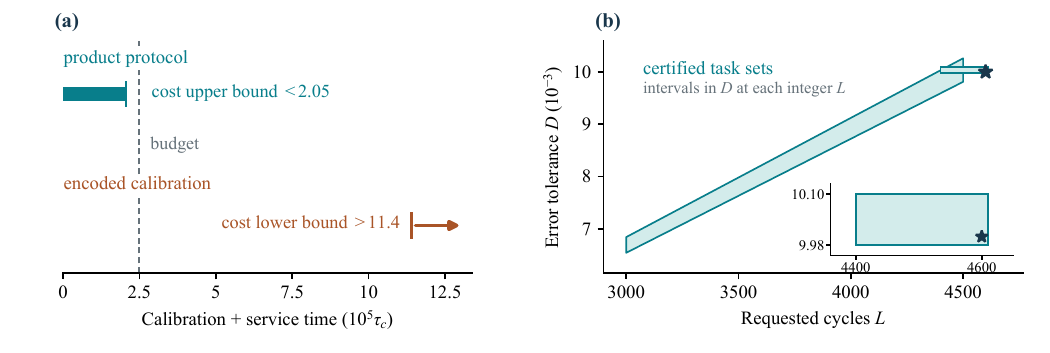}
\caption{Calibration preparation changes which tasks fit the budget.
(a) At the task starred in (b), the teal capped segment is the product-protocol total-cost upper bound, the orange rightward arrow is the encoded-calibration necessary lower bound, and the gray dashed line is the budget: $2.05<2.5<11.4$ in units of $10^5\tau_c$. Thus the product protocol is certified within budget while encoded calibration is excluded; the marks are bounds, not measured costs.
(b) Teal shading marks sufficient certified task sets for which the product protocol meets the target while encoded calibration is excluded at the same budget and deadline. The sloped band and rectangle give intervals of allowed error tolerance $D$ at each integer depth $L$; they are not uncertainty bands or optimized boundaries. The star is $(L,D)=(4600,0.01)$, and the inset enlarges the rectangle using the same $D$ units.
Both panels fix $\lambda=10^{-3}$, $n=10^4$, $r\le20$, $\zeta=0.02$, $B=2.5\times10^5\tau_c$, and a $5000\tau_c$ service deadline. Here $\zeta\tau_c$ is the duration of each added service pulse. Preparation, extraction, and delivered records are ideal; the stated pulse domain and appended phase fault are included in the bounds. The SM gives the certified domains and remaining working parameters.}
\label{fig:cost}
\end{figure*}

\emph{Complete encoded records versus a dedicated probe.---}
For the patch, the encoded comparator may prepare any logical/reference input
through the fixed encoding. Each query returns all eight checks and the
recovered logical output; actions and stopping may depend on prior records. A
label-independent simulation reduces the instrument to syndrome, pulse flag,
and a latent logical Pauli with label-independent processing. Even granting
that latent class and free external quantum memory, each weak-noise query has
KL at most $K_{\rm enc}\lambda^2$ for sufficiently small $\lambda$, uniformly
in the declared family \cite{prlv4supp}.
This full-output bound is specific to this instrument and is not a consequence
of Theorem~\ref{thm:matching} for arbitrary codes.

The dedicated probe resets the same data qubits to
$|+\rangle_1|+\rangle_2|0\rangle^{\otimes7}$ and uses the same extractor.
Two of the returned checks are $m_X=X_1X_2$ and
$m_Z=Z_1Z_2Z_4Z_5$. Put $c=(e^{-\lambda}+e^{-3\lambda})/2$ and
$d=(e^{-\lambda}-e^{-3\lambda})/2$. Their exact joint law is
\begin{equation}
 P_\theta(m_X=x,m_Z=z)=\frac{1+xc-\theta xzd}{4}.
 \label{eq:probe}
\end{equation}
Neither check alone reveals $\theta$; their correlation does. In model $+$,
the two $m_X=-1$ cells have leading probabilities $3\lambda/4$ and
$\lambda/4$, which swap in model $-$. Their information is
$\tfrac12\log(3)\lambda+O(\lambda^2)$ per probe. After $n$ trials, controlled
service is selected when
$T=N(-1,+1)-N(-1,-1)\ge0$ and idle otherwise.
This sign rule uses only observed counts, not $\lambda$; the calibrated noise
strength and model family enter the choice of $n$ and the task-cost proof.

\emph{From information to a guaranteed task.---}
For ideal swapped actions, let $R_{\rm oracle}^{\rm ideal}(L)$ be the error of
the better constant action when the label is known, and let
$R_{\rm blind}^{\rm ideal}(L)$ be the optimum over all randomized committed
binary words without that information. The SM gives both exactly
\cite{prlv4supp}. The condition
$R_{\rm oracle}^{\rm ideal}(L)\le D<R_{\rm blind}^{\rm ideal}(L)$ is the
ideal information-relevance window: learning can make an otherwise
unattainable target possible. At $\lambda=10^{-3}$ and $L=4600$,
the ideal optima are approximately $0.938\%$ and $1.105\%$, respectively,
placing $D=1\%$ between them. Preparation costs and imperfect pulses then
determine which tasks in the window admit a cost separation.

For the physical pulse model in the stated weak-noise regime,
\(0<q_{+,1}<q_{+,0}<1/2\) and \(0<q_{-,0}<q_{-,1}<1/2\). Write
$R_{\theta,a}=[1-(1-2q_{\theta,a})^L]/2$ and
\begin{equation}
 u=\frac{R_{+,0}-D}{R_{+,0}-R_{+,1}},\qquad
 v=\frac{D-R_{-,0}}{R_{-,1}-R_{-,0}}.
 \label{eq:uv}
\end{equation}
When $0<v<u<1$, data processing requires the pre-service transcript to carry
at least ${\rm kl}(u\Vert v)$ and ${\rm kl}(v\Vert u)$ in the two directions.
If $k_+,k_-$ bound the information per encoded query, the necessary query
count is at least \cite{kaufmann2016complexity}
\begin{equation}
 C_{\rm lb}=\max\!\left\{
 \frac{{\rm kl}(u\Vert v)}{k_+},
 \frac{{\rm kl}(v\Vert u)}{k_-}\right\}.
 \label{eq:costbound}
\end{equation}
Let $\tau_c$ be the duration of one idle, complete syndrome extraction, and
fixed recovery.
A complete product probe costs at most $r\tau_c$, including reset,
preparation, extraction, readout, record processing, final decision, and
cleanup; each of the two additional service pulses costs $\zeta\tau_c$.
For a product rule satisfying the target at every prefix and the service
deadline, one service call obeys
\begin{equation}
 \frac{B_{\rm prod}}{\tau_c}\le nr+L(1+2\zeta),\qquad
 \frac{B_{\rm enc}}{\tau_c}\ge C_{\rm lb}+L.
 \label{eq:total}
\end{equation}
A nonempty separating budget interval follows when
$nr+2L\zeta<C_{\rm lb}$. The encoded bound waives calibration preparation,
readout, and pulse overheads, plus external reference and memory costs.

For the finite example, the controlled action uses
$R_{x,1}(\pi+\alpha),R_{x,1}(\pi+\beta)$ bracketing the fixed idle channel,
with an independent $Z_1$ fault of probability $10^{-4}$ appended afterward.
Pulse durations are charged. At
$\lambda=10^{-3}$, $L=4600$, and $D=0.01$, $10^4$ product probes guarantee
error below $9.94\times10^{-3}$ uniformly over
$|\alpha|\le0.1$, $|\beta|\le\pi$. With $r\le20$ and $\zeta=0.02$, the
constructive total cost is below $2.05\times10^5\tau_c$, whereas encoded
calibration requires more than $1.14\times10^6\tau_c$ including service.
At the fixed $2.5\times10^5\tau_c$ budget, the displayed $n=10^4$ protocol
allows $r\le24.5216$; a larger $r$ requires a different separating budget.
Outward bounds certify the nonisolated Fig.~\ref{fig:cost}(b) task sets
\cite{prlv4supp}; the separation is not confined to one depth and tolerance.

Operational records require no additional probe cycles but may reveal too
little information in time. Online syndrome trackers address time-dependent
noise \cite{spitz2018adaptive,bhardwaj2026adaptive}; here the label is stable
and control is committed before service. For an $N$-cycle encoded record with
equal priors, the full-record cap and Pinsker's inequality yield
\begin{equation}
\begin{split}
 D_{\rm KL}(P_+^{Y_N}\Vert P_-^{Y_N})&\le Nk_+,\\
 p_{\rm err}^{\rm eq}&\ge\tfrac12[1-\sqrt{Nk_+/2}]_+ .
\end{split}
 \label{eq:observation}
\end{equation}
Here $[x]_+=\max\{x,0\}$. At the working point,
$k_+=1.30\times10^{-6}$ gives less than $6\times10^{-3}$ nats for
$N=L=4600$ and equal-prior identification error above $47\%$. This concerns
learning from records; a service-time feedback guarantee additionally needs a
joint analysis of acquired information and evolving logical error.

The separation persists as noise weakens. With ideal delivered records, set
$p_z=\lambda/10$ and keep the pulse domain and timing allowances fixed. A
family with
$L_\lambda=\Theta(\lambda^{-2})$,
$n_\lambda=\Theta(\lambda^{-1})$, and
$B_\lambda=b_0\tau_c/\lambda^2$, for positive $b_0$, has a strict full-cost
separation at $D=0.01$ for all sufficiently small $\lambda$, with deadline at
least $L_\lambda(1+2\zeta)\tau_c$ \cite{prlv4supp}. The proof fixes coefficients
that charge service on both sides; the information exponents alone do not
imply a divergent total-cost ratio.

Product calibration requires resolving the rare $m_X=-1$ sector. Independent
flips of the delivered $m_X,m_Z$ fields after ideal extraction and recovery,
never fed back into recovery, preserve first-order information when
$\epsilon_Z<1/2$ is fixed and $\epsilon_X=O(\lambda)$. Fixed
$0<\epsilon_X<1/2$ instead makes it second order \cite{prlv4supp}. At
$\epsilon_X=0.01$, $\epsilon_Z=0$, and the displayed $r=20$ cost, the present
budget admits too little total information for the same task guarantee; thus
Fig.~\ref{fig:cost} is an ideal-record certificate. A separate analytic
certificate uses $3\times10^5$ probes with $r\le3$ and $B=10^6\tau_c$ to
retain a strict separation. This record-flip model does not cover noisy
extraction, preparation errors, or noisy service recovery.

Reusing error-correction records avoids dedicated probe cycles, but useful
control also requires acquiring the relevant information in time. Here a
changed preparation exposes that information more often through the same
checks, and its charged overhead is repaid within the task budget. The
protected logical idle preserves quantum states and their correlations
between active gates, so the comparison matters for quantum computing as
well as storage. Calibration should therefore be judged by the protected
task it enables at total cost, not by learnability alone.

\acknowledgments
This work was supported by the Key-Area
Research and Development Program of Guangdong Province (2018B030326001),
the Science, Technology and Innovation Commission of Shenzhen Municipality
(JCYJ20170412152620376 and KYTDPT20181011104202253), the Innovation Program
for Quantum Science and Technology (2021ZD0301703), Guangdong Major Project
of Basic Research (2025B0303000007), the Shenzhen Science and Technology
Program (KQTD20200820113010023), and the CCF-QuantumCtek Superconducting
Quantum Computing Special Cooperation Program (CCF-QC2025002).

\paragraph*{Data and code availability.}
The code, interval certificates, and figure sources required to verify the
results accompany the submission and are available at
\url{https://github.com/QDynamics/CalibrationCostOfQControl}.

\begin{samepage}
\paragraph*{Author responsibility and tool disclosure.}
Generative-AI tools assisted with derivation checks, numerical cross-checks,
coding, language, and figure layout. The author retains full responsibility
for the scientific content.
\end{samepage}
\nocite{braunstein1994statistical}
\bibliography{references_PRLv6_Task_First_20260919}
\onecolumngrid
\clearpage
\begin{center}
{\bfseries Supplemental Material: Preparation Changes the Cost of Calibration for Quantum Control}\\[0.6ex]
{\small Xiu-Hao Deng}
\end{center}
\vspace{0.5\baselineskip}
\twocolumngrid
\setcounter{section}{0}
\setcounter{equation}{0}
\setcounter{theorem}{0}
\setcounter{lemma}{0}
\setcounter{proposition}{0}
\setcounter{table}{0}
\setcounter{figure}{0}
\setcounter{secnumdepth}{2}
\renewcommand{\theequation}{S\arabic{equation}}
\renewcommand{\thesection}{S\arabic{section}}
\renewcommand{\thetheorem}{S\arabic{theorem}}
\renewcommand{\thelemma}{S\arabic{lemma}}
\renewcommand{\theproposition}{S\arabic{proposition}}
\renewcommand{\thetable}{S\arabic{table}}
\noindent
The Letter compares two calibration routes for the same protected task.
Calibration precedes the unknown logical/reference input and selects a
committed service sequence; the total budget charges acquisition and
execution. This supplement establishes the control-relevant correlation,
the information available through each preparation, and the resulting
task-cost bounds. Section~\ref{sec:syndrome-criterion} proves the
general syndrome-information bound;
Sec.~\ref{sec:correlation-instance} specifies the patch, response mechanism,
and service channels. Sections~\ref{sec:matched-cost}
and \ref{sec:product-preparation} establish the full-record converse and
constructive product protocol; Sec.~\ref{app:weak-noise-proof} proves the
charged weak-noise separation. The remaining sections give mismatch
stability, certified task domains, and the delivered-readout condition.
The general theorem is syndrome-only; the full-quantum-output bound is
patch-specific. Finite guarantees use exact kernels and outward intervals.

\section{When syndrome records reveal a noise correlation}
\label{sec:syndrome-criterion}
Let $V$ exactly correct the coherent span of $I,Z_0,\ldots,Z_{n_{\rm q}-1}$,
as in the Knill--Laflamme conditions \cite{knill1997theory}.
In each cycle, centered Gaussian phases generate
$U(\phi)=\exp[-i\sum_j\phi_j Z_j/2]$, with covariance
$\lambda\Sigma_{\theta,a}$ on $n_{\rm q}$ physical qubits. The positive-definite shapes $\Sigma_{\theta,a}$
are held fixed as $\lambda\to0$, and $a$ belongs to a finite, classically
chosen action menu. Complete syndrome extraction and recovery are ideal;
noise is independent between cycles conditional on the stable label.
Actions and intervening trace-preserving logical operations depend only
on earlier syndromes, action choices, and label-independent private
randomness.

For each nonzero syndrome $s$ produced by a single $Z$ error, let $J_s$
contain the qubits producing that syndrome. Exact correction implies
$Z_jV=\nu_{s,j}Z_{j_s}V$ for $j\in J_s$, where $j_s$ is a representative
and $\nu_{s,j}=\pm1$. Extend $\boldsymbol\nu_s$ by zero outside $J_s$.
The leading syndrome probability is $\lambda\omega_{\theta,a,s}$, with
\begin{equation}
 \omega_{\theta,a,s}
 =\tfrac14\boldsymbol\nu_s^{\mathsf T}\Sigma_{\theta,a}
                       \boldsymbol\nu_s.
 \label{eq:syndrome-leading-rate}
\end{equation}
It is independent of the encoded input. The next theorem quantifies the
learning cost when these leading rates contain no model distinction.

\begin{theorem}[Syndrome-learning cost]
\label{thm:syndrome-matching}\hfill\\
In the setting above, suppose
\begin{equation}
 \omega_{+,a,s}=\omega_{-,a,s}
 \quad\text{for every admitted $a,s$.}
 \label{eq:syndrome-matching}
\end{equation}
The learning record consists of syndromes and classical action choices.
The logical state and an unmeasured reference may persist, with arbitrary
processing within the same encoded space between cycles. For sufficiently
small $\lambda$, every adaptive such record $Y$ with a hard cap of $N$
cycles obeys
\begin{equation}
 D_{\rm KL}(P_+^Y\Vert P_-^Y)\le K_{\rm syn}N\lambda^2,
 \label{eq:generic-syndrome-kl}
\end{equation}
and the reverse bound, for a constant $K_{\rm syn}$ depending on the fixed
code and covariance family. Identifying the label with each error probability
at most $\epsilon_{\rm id}\in(0,1/2)$ therefore requires
\begin{equation}
 N\ge\frac{{\rm kl}(1-\epsilon_{\rm id}\Vert\epsilon_{\rm id})}
              {K_{\rm syn}\lambda^2}.
 \label{eq:generic-syndrome-count}
\end{equation}
\end{theorem}

\label{app:syndrome-matching}
\subsection{Leading syndrome rates and uniform probabilities}
Exact correction of the coherent single-$Z$ error span implies
$V^\dagger Z_iZ_jV=\alpha_{ij}I$ by the Knill--Laflamme conditions
\cite{knill1997theory}. When the errors have the same syndrome,
$Z_iZ_j$ preserves the code. Its restriction is a Hermitian Pauli unitary,
so the scalar is $\alpha_{ij}=\pm1$. Thus
$Z_jV=\nu_{s,j}Z_{j_s}V$, as used in
Eq.~\eqref{eq:syndrome-leading-rate}. This conclusion uses coherent
correctability in addition to syndrome equality.

Write $\phi=\sqrt\lambda\,\xi$, where $\xi$ has the full-support Gaussian
law with covariance $\Sigma_{\theta,a}$. For an attainable syndrome $s$,
let $w_s$ be the minimum weight of a $Z$ Pauli producing it. The finite
product expansion of $U(\phi)$ gives the lowest-degree syndrome amplitude
as $\lambda^{w_s/2}M_s(\xi)V$, where
\begin{equation}
 M_s(\xi)V=
 \left(-\frac{i}{2}\right)^{w_s}
 \sum_{\substack{|E|=w_s\\s(E)=s}}
 \xi_E Z(E)V,\qquad \xi_E=\prod_{j\in E}\xi_j.
 \label{eq:generic-leading-amplitude}
\end{equation}
The distinct squarefree monomials $\xi_E$ have a strictly positive Gram
matrix under a full-support Gaussian: a nonzero polynomial cannot vanish
almost everywhere. If $g_s>0$ is its minimum eigenvalue, then for a unit
logical state $|\psi\rangle$,
\begin{equation}
 \mathbb E\|M_s(\xi)V|\psi\rangle\|^2
 \ge \frac{g_s}{4^{w_s}}
       \sum_{\substack{|E|=w_s\\s(E)=s}}
          \|Z(E)V|\psi\rangle\|^2>0.
 \label{eq:generic-gram-lower}
\end{equation}
Each norm in the sum equals one, making the lower bound uniform over the
logical state. Upper bounds follow from the finite Gram matrix norm.

For the Taylor remainder, express the syndrome probability as
$p_{\theta,a,s}(\rho)=\operatorname{Tr}[\rho B_{\theta,a,s}(\lambda)]$
on the finite logical space. Entries of the effect $B$ are finite sums
of Gaussian characteristic functions and are analytic in $\lambda$.
Gaussian parity removes odd total degrees in $\phi$. Consequently,
\begin{equation}
 B_{\theta,a,s}(\lambda)
 =\lambda^{w_s}\bigl[B_{\theta,a,s}^{(0)}+O(\lambda)\bigr],
 \quad B_{\theta,a,s}^{(0)}\ge c_s I>0,
 \label{eq:generic-effect-expansion}
\end{equation}
with an operator-norm remainder uniform over the finite model/action
family. Taking a sufficiently small common $\lambda_0$ ensures uniform
positive lower probability bounds. The operator inequality also tensors
with the reference identity; observed probabilities depend on the reduced
logical state. Unattainable syndromes have zero probability in both models
and can be omitted.

For $w_s=1$, correctability makes
$B_{\theta,a,s}^{(0)}=\omega_{\theta,a,s}I$. Under matching,
the probability difference is therefore $O(\lambda^2)$ even for a pair of
different logical states $\rho_+,\rho_-$. For $w_s\ge2$, both probabilities
are uniformly of order $\lambda^{w_s}$, as is an upper bound on their
difference. Normalization gives an $O(\lambda^2)$ difference in the
zero-syndrome probability, whose leading value is one. This also covers
codes with a weight-one $Z$ stabilizer: its zero syndrome is included
through normalization.

\subsection{Information and adaptive records}
Use the Kullback--Leibler (KL) bound $D_{\rm KL}(p\Vert q)\le\chi^2(p\Vert q)$
on the common support.
The zero-syndrome contribution is $O(\lambda^4)$. A one-error syndrome
contributes $O(\lambda^4)/\Theta(\lambda)=O(\lambda^3)$, and a sector of
minimum weight $w_s\ge2$ contributes $O(\lambda^{w_s})$.
There are finitely many sectors and actions, so, uniformly for any pair
of logical states,
\begin{equation}
 D_{\rm KL}\!\left(
 p_{+,a}(\cdot\mid\rho_+)\Vert p_{-,a}(\cdot\mid\rho_-)
 \right)\le K_{\rm syn}\lambda^2.
 \label{eq:generic-uniform-pair-kl}
\end{equation}
The reverse bound is identical. For a fixed classical history the
model-conditioned states can differ, and this is precisely the pair
covered by Eq.~\eqref{eq:generic-uniform-pair-kl}. Classical action rules
and private coins are label independent given that history. The KL chain
rule, with an early stop padded by null symbols to the hard cap $N$,
proves Eq.~\eqref{eq:generic-syndrome-kl}. Data processing to the binary
decision gives Eq.~\eqref{eq:generic-syndrome-count}.

If a fixed fresh input and action give distinct positive leading rates
in a minimum-weight-two syndrome, its binary occurrence law has KL and
squared Hellinger distance $\Theta(\lambda^2)$. Independent repetitions
attain the inverse-square testing order. Without that contrast,
information may be smaller: the theorem is not a universal tightness claim.

\subsection{A product probe with first-order information}
Prepare $|++\rangle$ on qubits $i,j$ with common individual variances and
different covariances $b_{\theta,p}=\Sigma_{\theta,ij}$.
The commuting $X_iX_j$ and $Z_iZ_j$ outcomes $x,z$ have exact law
\begin{equation}
 \begin{split}
 P_\theta(x,z)=\tfrac14\bigl[1
 &+x e^{-a_p\lambda}\cosh(b_{\theta,p}\lambda)\\
 &-xz e^{-a_p\lambda}\sinh(b_{\theta,p}\lambda)\bigr].
 \end{split}
 \label{eq:generic-product-law}
\end{equation}
Here $a_p=(\Sigma_{ii}+\Sigma_{jj})/2>|b_{\theta,p}|$.
Its $x=-1$ sector has probability
$\lambda(a_p+zb_{\theta,p})/4+O(\lambda^2)$. The two positive conditional $z$
laws differ by a fixed amount, in a sector of probability
$\Theta(\lambda)$. Its KL and squared Hellinger distance are therefore
$\Theta(\lambda)$ in both model directions. The common sector adds at
most $O(\lambda^2)$ KL. Product affinity under independent trials gives
an $O(\lambda^{-1})$ fixed-error test, and binary data processing with
the KL upper bound gives the matching necessary order. The readout may be entangling despite the product input; the patch
construction below uses the same extractor as encoded calibration.

\subsection{Why equal variances alone are insufficient}
The present surface code already provides a counterexample. Its
weight-two stabilizer $Z_0Z_3$ makes $Z_0$ and $Z_3$ share one nonzero
syndrome. Choose positive-definite covariance shapes
\begin{equation}
 \Sigma_\pm=I\pm\chi(e_0e_3^{\mathsf T}+e_3e_0^{\mathsf T}),
 \qquad 0<\chi<1.
 \label{eq:generic-degenerate-counterexample}
\end{equation}
Their diagonal entries coincide, but this syndrome has probability
$\lambda(1\pm\chi)/2+O(\lambda^2)$. Its binary occurrence law already
has first-order KL, with positive coefficient
\begin{equation}
 \frac{1+\chi}{2}\log\frac{1+\chi}{1-\chi}-\chi>0.
 \label{eq:generic-counterexample-kl}
\end{equation}
Thus, even a code correcting every single-qubit error can retain a first-order
correlation signal when errors share a syndrome.

For this patch the single-$Z$ groups are
$\{0,3\},\{1\},\{2\},\{4\},\{5,8\},\{6\},\{7\}$.
The two degenerate pairs have positive relative signs. With fixed
diagonals, matching imposes only
$\Delta\Sigma_{03}=\Delta\Sigma_{58}=0$. These are two independent
linear constraints on 36 off-diagonal entries. Their 34-dimensional
kernel describes first-order syndrome-invisible covariance directions,
with positive-definite perturbations available locally about an interior
covariance. A change in preferred logical control imposes further
conditions. The covariance pair of Eq.~\eqref{eq:qec-covariance} satisfies
both matching constraints, also after ideal $X_1$ conjugation.

\section{A correlation that changes the control decision}
\label{sec:correlation-instance}
\label{sec:logical}
Use the fixed $[[9,1,3]]$ rotated surface code
\cite{tomita2014surface}, with nine data qubits and eight reusable ancillas.
A storage cycle is a fixed idle followed by ideal complete syndrome
extraction and fixed recovery. The idle applies
\begin{equation}
 U(\phi)=\exp\!\left[-\frac{i}{2}\sum_{j=0}^{8}\phi_j Z_j\right],
 \qquad \phi\sim\mathcal N(0,K_\theta),
 \label{eq:qec-noise}
\end{equation}
where $\theta\in\{+,-\}$ is fixed throughout calibration and service, and
\begin{equation}
 \begin{gathered}
 K_\pm=\lambda(I+v_\pm v_\pm^{\mathsf T}),\qquad\lambda>0,\\
 v_+=(0,1,1,0,1,1,0,0,0)^{\mathsf T}.
 \end{gathered}
 \label{eq:qec-covariance}
\end{equation}
Only component 1 changes sign in $v_-$. Equivalently,
$\phi_j=\sqrt\lambda(\eta_j+v_{\theta,j}\xi)$ for independent standard-normal
local and common fluctuations. Signed common couplings describe correlated
dephasing \cite{yoshihara2010correlated,vonlupke2020spectroscopy}.
The known $\lambda$ sets the noise strength; the unknown model label fixes the
correlation sign through calibration and service. Conditional on that label,
all cycles are independent. The individual phase laws and $|K_{ij}|$ agree.

The two actions are no pulse ($0$), or $X_1$ conjugation bracketing the idle
($1$). Calibration selects a sequence in this binary control menu before
service. During the protected operation, syndromes drive the fixed recovery
but do not update that control sequence. Retaining them for noise inference
is considered separately in Sec.~\ref{app:task-cost-proof}.

For encoding $V$, syndrome projector $\Pi_s$, and recovery $R_s$,
\begin{equation}
 V^\dagger R_s\Pi_sU(\phi)V
 =\gamma_{s,0}(\phi)I+\gamma_{s,1}(\phi)Z_L.
 \label{eq:recovered-amplitude}
\end{equation}
The parity argument below removes the averaged cross term. Thus
$\mathcal E_q(\rho)=(1-q)\rho+qZ_L\rho Z_L$, with
$q=\sum_s\mathbb E|\gamma_{s,1}|^2$, has half-diamond error $q$.

All single-$Z$ errors are corrected, so the leading logical risk is governed
by a fourth-order phase moment:
\begin{equation}
 Q_4(K)=\frac1{16}\sum_s
 \sum_{\substack{(i,j)\in\mathcal B_s\\(k,l)\in\mathcal B_s}}
 (K_{ij}K_{kl}+K_{ik}K_{jl}+K_{il}K_{jk}),
 \label{eq:wick-risk}
\end{equation}
where $\mathcal B_s$ contains pairs whose recovery leaves logical $Z$.
Ideal conjugation swaps the two logical risks. The exact kernel below
bounds the finite-noise rates and verifies isolated-cycle recovery optimality.

A committed word is a control sequence chosen before service and held fixed
during it. Conditional independence gives its exact channel error
\begin{equation}
q_\theta(a_1\cdots a_L)=\tfrac12\left[1-\prod_{j=1}^L(1-2q_{\theta,a_j})\right].
\label{eq:qec-composition}
\end{equation}
This identity includes arbitrary reference-entangled inputs.

\subsection{Recovery kernel and physical response}
\label{app:qec-kernel}
Integer masks use little-endian qubit indices $0,\ldots,8$.
The eight checks in order are
\[
 \begin{gathered}
 X(27),\; Z(54),\; Z(216),\; X(432),\\
 X(192),\; X(6),\; Z(9),\; Z(288).
 \end{gathered}
\]
Here $X(m)$ or $Z(m)$ acts on the set bits of $m$. Logical operators are
$X_L=X(73)$ and $Z_L=Z(7)$. Ideal extraction uses 24 controlled-NOT (CNOT) gates, eight Hadamards,
and eight ancilla measurements/resets; these operations are part of the
nominal cycle, not assumed to take zero time. Their noise is ideal in this
example. Encoding is also ideal and outside the storage task, whose input
is an encoded unknown qubit. For each syndrome, choose a minimum-weight
$Z$ mask, breaking ties by integer value. A nonzero $X_1$ flag is first
corrected with $X_1$. All remaining $Z$ recoveries are independent of that flag.

Expand the idle as
$U(\phi)=\sum_E a_E(\phi)Z(E)$, where
\[
 a_E=\prod_{i\in E}[-i\sin(\phi_i/2)]
      \prod_{i\notin E}\cos(\phi_i/2).
\]
For errors of syndrome $s$, let $\ell(E)\in\{0,1\}$ denote the logical
$Z$ class of $R_sZ(E)$. Then
\[
 \gamma_{s,0}=\sum_{s(E)=s,\,\ell(E)=0}a_E,\qquad
 \gamma_{s,1}=\sum_{s(E)=s,\,\ell(E)=1}a_E .
\]
All $Z$ stabilizers have even weight and logical $Z$ has odd weight.
Since $a_E(-\phi)=(-1)^{|E|}a_E(\phi)$, inversion symmetry gives
$\mathbb E\gamma_{s,0}\gamma_{s,1}^*=0$. The syndrome-resolved channel is
\[
 \begin{gathered}
 \mathcal C_s(\rho)=p_{s0}\rho+p_{s1}Z_L\rho Z_L,\\
 p_{s0}=\mathbb E|\gamma_{s,0}|^2,\qquad p_{s1}=\mathbb E|\gamma_{s,1}|^2 .
 \end{gathered}
\]
The sum of all weights is one. No Pauli twirl has been inserted.

Correctability removes constant and single-angle terms from
$\gamma_{s,1}$. Its leading even-parity amplitude is
$-\frac14\sum_{(i,j)\in\mathcal B_s}\phi_i\phi_j$; Wick averaging yields
$q=Q_4(K)+O(\lambda^3)$. The 18 malignant pairs among the 36 pairs form
seven syndrome groups, listed in Eq.~\eqref{eq:patch-malignant-pairs}.

\subsubsection{Quartic path interference and a syndrome-response check}
For ideal $X_1$ conjugation and no appended fault, write
$p_{\theta,a}(s,z)$ for the probability of phase syndrome $s$ and recovered
logical class $z\in\{0,1\}$, and put
\begin{equation}
 \Delta p_\theta(s,z)=p_{\theta,0}(s,z)-p_{\theta,1}(s,z),\quad
 \Delta q_\theta=q_{\theta,0}-q_{\theta,1}.
 \label{eq:response-differences}
\end{equation}
All differences are idle minus conjugated idle with the known frame restored.

Within-block interference requires
$E\mathbin{\mathrm{xor}}F$ to be a $Z$ stabilizer, and a conjugation-odd
term requires that stabilizer to contain site 1. No weight-two stabilizer
does; the only fourth-degree supports are
\begin{equation}
 \begin{aligned}
 S_1&=\{1,2,4,5\} &&\text{(mask 54)},\\
 S_2&=\{1,2,4,8\} &&\text{(mask 278)}.
 \end{aligned}
 \label{eq:response-supports}
\end{equation}
For an ordered disjoint partition $E\mathbin{\dot\cup}F=S_r$, the leading
amplitude product is
$(-1)^{|E|}\prod_{j\in S_r}\phi_j/16$.  Reversing $\phi_1$ changes its sign,
so the idle-minus-conjugated difference doubles this contribution once, from
$1/16$ to $1/8$ for each ordered pair.  Define
\begin{equation}
 M_\theta=\mathbb E_\theta[\phi_1\phi_2\phi_4(\phi_5+\phi_8)].
 \label{eq:response-moment}
\end{equation}
Assigning the sixteen ordered partitions of each support to their common
syndrome and logical class gives
\begin{equation}
 \Delta p_\theta(s,z)=c_{s,z}\frac{M_\theta}{4}+O(\lambda^3),
 \label{eq:response-block-expansion}
\end{equation}
where the complete table of nonzero fourth-order coefficients is
\begin{equation}
\begin{array}{c|c}
 (s,z)&c_{s,z}\\ \hline
 (0,0),(1,1),(10,0),(11,0)&+1\\
 (2,0),(3,0),(8,0),(9,0)&-1
\end{array}
 \label{eq:response-block-table}
\end{equation}
All other fourth-order block responses vanish, and the displayed entries sum
to zero as normalization requires.  The syndrome integers use the ordered
$X$ checks $X(27),X(432),X(192),X(6)$, with the first check in the least
significant bit; no translation from another bit convention is being used.

Summing the logical-error blocks, and separately summing both logical
classes of syndromes 10 and 11, proves
\begin{equation}
 \begin{aligned}
 \Delta q_\theta
 &=\frac{M_\theta}{4}+O(\lambda^3)\\
 &=\frac{\Delta P_{\theta}(10)+\Delta P_{\theta}(11)}{2}
   +O(\lambda^3),
 \end{aligned}
 \label{eq:patch-syndrome-response}
\end{equation}
where $\Delta P_\theta(s)=\sum_z\Delta p_\theta(s,z)$.  For the Gaussian
local-plus-common family in Eq.~\eqref{eq:qec-covariance}, with
$\theta\in\{+1,-1\}$ and $v_{\theta,1}=\theta$, the distinct-index moment is
$M_\theta=3\theta\lambda^2$.  Hence
\begin{equation}
 q_{\theta,0}-q_{\theta,1}
 =\frac{3\theta}{4}\lambda^2+O(\lambda^3).
 \label{eq:patch-gaussian-response}
\end{equation}
The $\phi_8$ moment vanishes for this coupling vector, but belongs to the
structural rule. Wick's rule enters only the Gaussian simplification.
The logical response occurs in $(s,z)=(1,1)$; syndromes 10 and 11 are its
observable proxies, with leading probabilities $5\lambda^2/4$ and
$\lambda^2/2$ swapped by action or label. Their A/B counts therefore need
$\Theta(\lambda^{-2})$ independent cycles for a fixed-error decision.

\paragraph*{Scoped one-logical-qubit CSS selection rule.}
For completeness, the same path argument has the following limited extension.
Consider a Calderbank--Shor--Steane (CSS) encoding of one logical qubit that
(i) corrects every single
$Z$ error; (ii) uses a fixed minimum-weight $Z$ recovery with a declared tie
rule; and (iii) is driven by
$\phi=\sqrt\lambda\,\boldsymbol\xi$, where the law of
$\boldsymbol\xi$ is independent of $\lambda$, globally inversion symmetric,
and has finite $\mathbb E\|\boldsymbol\xi\|_1^6$.  Fix a conjugated site $k$
and assume (iv) no weight-two element of the $Z$-stabilizer group
$\mathcal S_Z$ contains $k$; and (v) every element of $\mathcal S_Z$ has even
weight while the chosen logical $Z_L$ has odd weight.  We use the standard
positive-sign convention for the represented $Z$ stabilizer masks, so each
$Z(S)$ with $S\in\mathcal S_Z$ acts as $+I$ on the code space.  Then the
conjugation-odd response of every retained syndrome-and-class probability is
independent of the logical input, has no contribution below degree four, and its degree-four
part is as follows.  Put
\begin{equation}
 m_{S,E}=(-1)^{|E|}
 \mathbb E\!\left[\prod_{j\in S}\phi_j\right].
 \label{eq:scoped-css-moment}
\end{equation}
Then
\begin{equation}
\begin{aligned}
 \Delta P(s,z)&=\frac18
 \sum_{\substack{S\in\mathcal S_Z,\ |S|=4,\\ k\in S}}
 \sum_{\substack{E\mathbin{\dot\cup}F=S\\\text{ordered}}}m_{S,E}\\
 &\quad\times
 \mathbf 1\!\left\{(s(E),\ell(E))=(s,z)\right\}
 +O(\lambda^3).
 \label{eq:scoped-css-quartic-rule}
\end{aligned}
\end{equation}
Indeed, class-preserving interference requires
$E\mathbin{\mathrm{xor}}F\in\mathcal S_Z$.  Condition (iv) excludes a
degree-two conjugation-odd term, while condition (v) excludes odd stabilizer
supports altogether.  At degree four the leading amplitude product is
$m_{S,E}/16$ per ordered pair.  The idle-minus-conjugated sign change doubles
that ordered-pair contribution once, yielding the factor $1/8$ in
Eq.~\eqref{eq:scoped-css-quartic-rule}.  Condition (v) makes the two
recovered logical classes have opposite inversion parity, removing their
cross term and the logical-input dependence.  Taylor's theorem and the finite
sixth moment bound the averaged remainder by $O(\lambda^3)$.

Allowed moments can vanish or cancel, and a block response need not change
logical error. A syndrome-count representation requires a separate
code/recovery identity; here its weight is
$\frac12\mathbf1_{\{10,11\}}$. A weight-two stabilizer at the conjugated site
or an even-weight logical $Z$ violates the stated hypotheses. The rule alone
does not establish a learning advantage.

\subsubsection{Finite-parameter bounds}
Let $\mathcal O_0$ be the 16 bit strings generated from zero by the four
$X$ checks, and $\mathcal O_1=\mathcal O_0\mathbin{\mathrm{xor}}73$.
The encoded basis states are uniform superpositions on these disjoint
orbits. On their union, define
\[
 a_s(z)=\frac{(-1)^{R_s\cdot z}}{32},\qquad
 b_s(z)=a_s(z)
 \begin{cases}1,&z\in\mathcal O_0,\\-1,&z\in\mathcal O_1.\end{cases}
\]
With $e_i(z)=(-1)^{z_i}$, put
\[
 d_{zz'}=\frac{[e(z)-e(z')]^{\mathsf T}K[e(z)-e(z')]}{8},
 \qquad G_{zz'}=e^{-d_{zz'}} .
\]
Gaussian averaging gives $p_{s0}=a_s^{\mathsf T}Ga_s$ and
$p_{s1}=b_s^{\mathsf T}Gb_s$. An appended $Z_1$ multiplies both vectors
by $(-1)^{z_1}$. Thus all four fault/no-fault risks can be evaluated on a
32-dimensional kernel rather than a sampled phase ensemble.

At rational covariance points each risk or margin is
$\sum_d h_d e^{-d}$ with rational $h_d$. For $d\ge0$,
\[
 \left|e^{-d}-\sum_{j=0}^4\frac{(-d)^j}{j!}\right|\le d^5/120.
\]
Grouping equal $d$ before applying this remainder retains the cancellations.
The certificate stores outward endpoints on denominator $10^{15}$.
Alternating bounds for $\sin(0.05)$ through degrees seven and nine bound
the pulse mixture. The full pulse-box rates are given once, in
Eq.~\eqref{eq:full-box-service-rates}.

\subsubsection{Why a weak single-cycle decoder is not responsible}
For a syndrome-resolved $I/Z$ channel, the best deterministic completely positive
trace-preserving (CPTP) recovery
has error
\begin{equation}
 q_{\rm opt}=\sum_s\min(p_{s0},p_{s1}).
 \label{eq:one-cycle-optimum}
\end{equation}
To see the lower bound, use a Bell input. If $T_j$ are the logical Kraus
operators of an arbitrary recovery after syndrome $s$, its contribution
to entanglement fidelity is
$\sum_{j,l}p_{sl}|\operatorname{Tr}(T_jZ^l)|^2/4$.
Pauli orthogonality and $\sum_jT_j^\dagger T_j=I$ bound this by
$\max_l p_{sl}$. A Bell measurement therefore lower-bounds half-diamond
error by Eq.~\eqref{eq:one-cycle-optimum}. Choosing the more probable
logical Pauli for each syndrome attains the bound.

Rational kernel intervals give $p_{s0}>p_{s1}$ for all 16 syndromes in
both ideal models, with margins exceeding $1.24\times10^{-7}$.
The fixed recovery therefore attains this isolated-cycle optimum.
The multi-cycle comparison still fixes the extraction and recovery;
it does not optimize arbitrary multi-round decoders.

\subsection{Optimal control sequences with and without the noise label}
The blind comparator must include every classical randomization over complete
control words, rather than only the two deterministic constant words.  Put
$s(q)=1-2q$ and
\[
 a=s(q_{+,0}),\quad b=s(q_{+,1}),\quad
 c=s(q_{-,0}),\quad d=s(q_{-,1}).
\]
Assume $b>a$, $c>d$, $c>a$, and $b>d$, with all four factors in $(0,1)$,
and write $a_L=a^L$, $b_L=b^L$, $c_L=c^L$, and $d_L=d^L$.  If the label is
unknown before service, the minimum worst-label error over all distributions
on length-$L$ committed binary words is
\begin{equation}
 R_{\rm blind}(L)=\frac12\left[
 1-\frac{b_Lc_L-a_Ld_L}{b_L-a_L+c_L-d_L}\right].
 \label{eq:matched-optimum}
\end{equation}
It is attained by drawing the constant word $1^L$ once per program with
probability
\begin{equation}
 t_*=\frac{c_L-a_L}{b_L-a_L+c_L-d_L}
 \label{eq:matched-weight}
\end{equation}
and drawing $0^L$ otherwise.  To prove optimality, consider a deterministic
word with $m$ action-1 cycles.  Its two label-conditioned survivals are
$a^{L-m}b^m$ and $c^{L-m}d^m$, each convex in $m$.  Replacing the word by the
endpoint mixture with weights $1-m/L$ and $m/L$ therefore increases both
survivals.  Applying this replacement inside any word distribution reduces
the minimax problem to endpoint mixtures.  Their survivals are
$a_L+t(b_L-a_L)$ and $c_L+t(d_L-c_L)$; one increases and the other decreases,
and the stated order assumptions place their intersection at $t_*\in(0,1)$.
Maximizing the smaller survival proves Eqs.~\eqref{eq:matched-optimum} and
\eqref{eq:matched-weight}.  Because every factor lies below one, the same
mixture's final-depth error bound also protects every earlier prefix.

For ideal swapped actions, $a=d=x<b=c=y$, this reduces to
\begin{equation}
 R_{\rm blind}^{\rm ideal}(L)=\frac{2-x^L-y^L}{4},\qquad
 R_{\rm oracle}^{\rm ideal}(L)=\frac{1-y^L}{2}.
 \label{eq:ideal-optima}
\end{equation}
The oracle knows the label. At $\lambda=10^{-3}$ and $L=4600$, the
ideal kernel gives $R_{\rm oracle}^{\rm ideal}=0.009377084969\ldots$
and $R_{\rm blind}^{\rm ideal}=0.011053304716\ldots$, bracketing $D=0.01$.
Imperfect pulses use the four actual rates in
Eq.~\eqref{eq:matched-optimum}; costs and deadline further constrain feasibility.

\subsection{Imperfect pulses at the recovery interface}
We model actual pulses as $R_{x,1}(\pi+\alpha)$ and
$R_{x,1}(\pi+\beta)$, and an independent $Z_1$ fault of probability $p_z$
after the pair. The appended fault belongs to the imperfect controlled action;
the idle action has no appended pulse fault. These imperfections are shared by
the blind and informed controllers. Complete ideal extraction distinguishes the single-$X_1$
flag, corrects it, and applies the same $Z$ recovery irrespective of that
flag.

\begin{proposition}[Recovery-filtered pulse error]
\label{prop:pulse}
For diagonal idle Kraus operators and this CSS recovery interface, the
logical channel after marginalizing the $X_1$ flag obeys
\begin{equation}
 \mathcal C_{\alpha,\beta}
 =\cos^2(\alpha/2)\mathcal C_1+\sin^2(\alpha/2)\mathcal C_0,
 \label{eq:pulse-identity}
\end{equation}
where $\mathcal C_1$ is ideal conjugation and $\mathcal C_0$ is no pulse.
This holds for arbitrary logical/reference inputs, and separately for each
retained $X$-check syndrome.
\end{proposition}
The following proof is exact in both angles. Put $p_x=\sin^2(\alpha/2)$;
for an independent random $\alpha$, use its expectation instead.
The factorization assumes the pulses bracket an unchanged idle channel.

The appended $Z_1$ fault gives
\begin{equation}
\begin{split}
 q_{{\rm ctrl},+}={}&(1-p_z)[(1-p_x)q_-+p_xq_+]\\
 &+p_z[(1-p_x)q_-^Z+p_xq_+^Z],
 \label{eq:pulse-risk}
\end{split}
\end{equation}
with signs exchanged for model $-$. Here $q_\pm^Z$ are the recovered risks
with the extra fault. Since that fault alone is correctable,
$q_\pm^Z=O(\lambda)$, giving a leading contribution of order $p_z\lambda$.
Both sides use the same imperfect actions;
Eq.~\eqref{eq:full-box-service-rates} gives the finite working-point rates.

\paragraph*{Proof.}
\label{app:pulse}
Let $N_\mu$ be a diagonal idle Kraus operator, $X=X_1$,
$W_{1,\mu}=XN_\mu X$, and $W_{0,\mu}=N_\mu$. Up to a global sign, the physical
operator for the two overrotated pulses is
$R_x(\beta)W_{1,\mu} R_x(\alpha)$.
Write $c_\alpha=\cos(\alpha/2)$ and $s_\alpha=\sin(\alpha/2)$,
and similarly for $\beta$. Its even-$X$ part is
$c_\beta c_\alpha W_{1,\mu}-s_\beta s_\alpha W_{0,\mu}$.
The odd part is in a distinct $Z$-check syndrome sector and, after
correcting $X$, becomes
$-i(s_\beta c_\alpha W_{1,\mu}+c_\beta s_\alpha W_{0,\mu})$.
After applying the common $X$-check projection and $Z$ recovery, call
their logical images $W_{1,s\mu},W_{0,s\mu}$. The two flag-branch Kraus
operators are
\begin{align*}
 L_0&=c_\beta c_\alpha W_{1,s\mu}-s_\beta s_\alpha W_{0,s\mu},\\
 L_1&=-i(s_\beta c_\alpha W_{1,s\mu}+c_\beta s_\alpha W_{0,s\mu}).
\end{align*}
In $\sum_{j=0}^1L_j\rho L_j^\dagger$ the cross coefficients cancel.
The remaining coefficients are $c_\alpha^2$ and $s_\alpha^2$,
proving Eq.~\eqref{eq:pulse-identity} for every reference input.
Summation over $\mu$ and, if desired, $s$ preserves the equality.

Appending the stochastic $Z_1$ gives Eq.~\eqref{eq:pulse-risk};
commuting $X_1$ and $Z_1$ only introduces a branchwise global sign.
For centered Gaussian phases, the parity argument still eliminates the
logical cross term after this fixed extra Pauli. The extra fault alone
is corrected, so a malignant amplitude requires at least one remaining
phase factor and its risk starts at $O(\lambda)$.

The equality marginalizes the flag: even with ideal idle and $\alpha=0$,
$\beta=0$ and $\beta=\pi$ have identical recovered channels but different
flags. Noise-correlated $\alpha$ need not give one independent mixture weight.
Theorem~\ref{thm:task-calibration-cost} does not use this equality to discard
the flag. Its calibration instrument retains the resulting beta-dependent
within-flag interference and bounds the complete record directly.

\section{Full-record information and the task-cost bound}
\label{sec:matched-cost}
An unknown logical/reference input arrives after calibration; conditional on
the stable label, it is independent of the service noise. The transcript selects a
committed binary word; service syndromes drive only the fixed recovery.
Error is $\frac12\lVert\overline{\mathcal C}-\mathcal I\rVert_\diamond$,
the half-diamond distance of the calibration-averaged service channel from the
ideal idle, not a per-transcript confidence guarantee.

Each encoded query consists of exactly one fixed-duration noise/recovery cycle.
Its input is any logical/reference state through the fixed encoding, and
it returns the recovered logical output and all eight check bits.
Label-independent logical processing and re-encoding, classical action
choices, and classical stopping may depend on past records. An ideal
external reference or memory and the latent logical Pauli are granted at
no cost to strengthen the converse. Extra cycles count as extra queries.
The action interfaces are not coherently superposed, and sacrificial
out-of-code preparations define the separate product interface.

For each action,
\begin{equation}
\mathcal C_{\theta,a}(\rho)=(1-q_{\theta,a})\rho+q_{\theta,a}Z_L\rho Z_L,
\quad s_{\theta,a}=1-2q_{\theta,a}.
\label{eq:task-recovered-cycle}
\end{equation}
Let $R_{\theta,a}(L)=[1-s_{\theta,a}^L]/2$. Define
\begin{equation}
\begin{gathered}
b_+=R_{+,0},\quad g_+=R_{+,1},\quad
g_-=R_{-,0},\quad b_-=R_{-,1},\\
u=(b_+-D)/(b_+-g_+),\qquad v=(D-g_-)/(b_--g_-).
\end{gathered}
\label{eq:task-asymmetry}
\end{equation}
\begin{theorem}[Task-information cost bound]
\label{thm:task-calibration-cost}\hfill\\
Assume the recovered cycles in Eq.~\eqref{eq:task-recovered-cycle} have
positive dephasing contractions and opposing action preferences:
\(0<s_{+,0}<s_{+,1}<1\) and \(0<s_{-,1}<s_{-,0}<1\).
Assume a stable label \(\theta\in\{+,-\}\), fresh service noise conditional on the label,
independence between calibration and service, and conditionally independent
service cycles. A possibly adaptive calibration transcript \(Y\) may select
any distribution on length-\(L\) committed binary control words. If its
calibration-averaged service channel has half-diamond error at most
\(D\) in both models at every prefix, and \(0\le v<u\le1\), then
\begin{equation}
\begin{split}
 D_{\rm KL}(P_+^Y\Vert P_-^Y)&\ge {\rm kl}(u\Vert v),\\
 D_{\rm KL}(P_-^Y\Vert P_+^Y)&\ge {\rm kl}(v\Vert u),
\end{split}
 \label{eq:task-kl-requirement}
\end{equation}
where \({\rm kl}\) is binary relative entropy. Suppose further that, after
every adaptive history and for either next action, each charged calibration
query contributes at most \(k_+\) and \(k_-\) nats in the forward and reverse
KL directions. Every such calibration with hard query cap \(C\) obeys
\begin{equation}
 C\ge C_{\rm lb}(L,D):=
 \max\!\left\{\frac{{\rm kl}(u\Vert v)}{k_+},
               \frac{{\rm kl}(v\Vert u)}{k_-}\right\}.
 \label{eq:task-direct-cost}
\end{equation}
The subscript \({\rm lb}\) denotes a necessary lower bound on the query
count. Dependence on the specified noise family, pulse parameters, and
admitted query interface is suppressed.
For one service call, suppose a constructive physical-preparation protocol
satisfies the same all-prefix error bound and the service deadline. If its
\(n\) complete probes, including processing and cleanup, cost at most
\(nr\tau_c\), and each additional service
pulse has duration \(t_\pi=\zeta\tau_c\), its total cost and the encoded-query
necessary cost obey
\begin{align}
 \frac{B_{\rm prep}}{\tau_c}&\le nr+L(1+2\zeta),\nonumber\\
 \frac{B_{\rm encoded}}{\tau_c}&\ge C_{\rm lb}(L,D)+L.
 \label{eq:matched-total-cost}
\end{align}
\end{theorem}

The finite working point is $\lambda=10^{-3}$, $p_z=10^{-4}$, and
$\mathcal P=[-0.1,0.1]\times[-\pi,\pi]$ for pulse angles
$(\alpha,\beta)$. Preparation, readout, extraction, and fixed recovery
are ideal. Each pulse is modeled as a rotation bracketing a fixed idle
channel. Its charged duration adds overhead but does not change that
idle channel. The angles are granted as known to the encoded comparator,
which strengthens its lower bound but does not enlarge the frozen physical
action menu; all statements are uniform over the box. The product sign rule
does not use $\lambda$ or these angles, although choosing $n$ and certifying
the task window use the specified model family and noise strength.
\label{app:task-cost-proof}
All logarithms are natural; initial calibration states and private
randomness are label independent.

\subsection{From program risk to asymmetric information}
At transcript $Y=y$, a word with $m$ action-1 cycles has survival
$S_\theta(m)=s_{\theta,0}^{L-m}s_{\theta,1}^m$. Convexity gives
\begin{equation}
 S_\theta(m)\le(1-m/L)s_{\theta,0}^{L}+(m/L)s_{\theta,1}^{L}.
 \label{eq:task-endpoint-chord}
\end{equation}
As in the blind optimum, replace the conditional word distribution by
$1^L$ with probability $w(y)=\mathbb E[m/L\mid Y=y]$ and $0^L$ otherwise.
This improves both terminal errors, including reference inputs. Each
endpoint-word error increases with depth, so the replacement also satisfies
every prefix whenever its terminal error is feasible.
With
$t_\theta=\mathbb E_\theta w(Y)$, the two risk constraints are
\begin{align}
 (1-t_+)b_+ +t_+g_+&\le D,\nonumber\\
 (1-t_-)g_- +t_-b_-&\le D,
 \label{eq:task-endpoint-risks}
\end{align}
and hence $t_+\ge u$ and $t_-\le v$.

Append an independent uniform random variable $U$ and define
$J=\boldsymbol 1\{U\le w(Y)\}$. Then $J$ is Bernoulli with parameter
$t_\theta$. Classical data processing gives
\begin{align}
 D_{\rm KL}(P_+^Y\Vert P_-^Y)
 &\ge {\rm kl}(t_+\Vert t_-),\nonumber\\
 D_{\rm KL}(P_-^Y\Vert P_+^Y)
 &\ge {\rm kl}(t_-\Vert t_+).
 \label{eq:task-data-processing}
\end{align}
For $x>y$, $\mathrm{kl}(x\Vert y)$ increases with $x$ and decreases with
$y$; the reverse divergence has the corresponding reversed monotonicities.
The minima over $t_+\ge u$, $t_-\le v$ occur at $(u,v)$, proving
Eq.~\eqref{eq:task-kl-requirement}, including endpoints by continuity.
If $u\le v$, a transcript-independent endpoint mixture may remove the positive
information requirement. If $u>1$ or $v<0$, a best constant action already
fails in one model and the chord proves infeasibility within this word class.
Degenerate equal-action rates are evaluated directly rather than through a
zero denominator.

Equation~\eqref{eq:task-endpoint-chord} is a risk domination only. A mixture
with small positive $w$ can introduce an all-pulse branch and need not reduce
worst-case pulse time. We do not use it as a service-cost optimum.

\subsection{The complete recovery record}
Write the two physical paths of the controlled action as
$W_1(\phi)=X_1U(\phi)X_1$ and $W_0(\phi)=U(\phi)$. For a four-bit phase syndrome
$s$, logical class $z\in\{0,1\}$, and independent appended-$Z_1$ fault branch
$e$, let $f_{1,sz,e}$ and $f_{0,sz,e}$ be the recovered logical-Pauli
amplitudes. The appended branch is present only for the controlled action;
the idle action has no pulse fault.

The even-weight stabilizers and odd-weight logical $Z_L$ give opposite
inversion parities to the two recovered classes. Both paths and their real
flag combinations preserve this distinction; an appended $Z_1$ shifts both
parities together. Hence the logical cross term is odd and averages to zero
separately for every syndrome, flag, and fault branch. The retained instrument is
\begin{equation}
 \mathcal I_{\theta,a}(\rho)
 =\sum_{s,b,z}P_{\theta,a}(s,b,z)
 |s,b\rangle\!\langle s,b|\otimes Z_L^z\rho Z_L^z.
 \label{eq:b2-recovered-instrument}
\end{equation}
Here $s$ gives the four $X$-check outcomes. Of the four $Z$ checks
$Z(54),Z(216),Z(9),Z(288)$, the corrected $X_1$ path flips only $Z(54)$.
Thus its flag $b$ is one of the eight stabilizer bits, while the other three
$Z$-check outcomes are deterministic; $b$ is not a ninth independent bit.

Equation~\eqref{eq:b2-recovered-instrument} holds for every encoded
logical/reference input, not only for a diagnostic state. Treating the latent
$z$ as observed makes each conditional logical output a known Pauli. Any retained
reference or quantum memory, logical preparation, processing, and readout can
then be simulated by a label-independent map after the classical draw
$(s,b,z)$. The latent bit and external memory are free relaxations of the converse;
neither is claimed as a component of the physical 17-qubit construction.

The retained flag mixes the coherent paths, so its information requires
the within-class Gram blocks.
For $j=(s,z)$, define the recovery-resolved two-path Gram block
\begin{equation}
 G_{\theta,j}=\sum_e w_e\,\mathbb E_\theta
 \begin{pmatrix}f_{1,j,e}\\f_{0,j,e}\end{pmatrix}
 \begin{pmatrix}f_{1,j,e}^*&f_{0,j,e}^*\end{pmatrix}
 =\begin{pmatrix}a_{\theta,j}&c_{\theta,j}\\
                  c_{\theta,j}&d_{\theta,j}\end{pmatrix},
 \label{eq:full-record-gram}
\end{equation}
where $w_1=p_z$ for the controlled action. Thus $a_{\theta,j}$ is the
conjugated-path weight (the $W_1=X_1UX_1$ entry $f_{1,j,e}$) and
$d_{\theta,j}$ the idle-path weight ($W_0=U$), and the label $\theta=+$ means
$v_+$ in Eq.~\eqref{eq:qec-covariance}. The blocks are real, positive
semidefinite. Exact interval arithmetic below verifies strict positivity of
all 64 model-resolved blocks at the finite working point, with determinant lower
endpoints greater than $7.72\times10^{-20}$.

Set
\begin{equation}
 D_\alpha=\operatorname{diag}(\cos(\alpha/2),\sin(\alpha/2)),
 \qquad \Gamma_{\theta,j}=D_\alpha G_{\theta,j}D_\alpha,
 \label{eq:full-record-gamma}
\end{equation}
and introduce the orthonormal flag vectors
\begin{align}
 u_0&=(\cos(\beta/2),-\sin(\beta/2))^{\mathsf T},\nonumber\\
 u_1&=(\sin(\beta/2), \phantom{-}\cos(\beta/2))^{\mathsf T}.
 \label{eq:full-record-flag-vectors}
\end{align}
Then the complete controlled-action probabilities are
\begin{equation}
 P_{\theta,1}(j,b)=u_b^{\mathsf T}\Gamma_{\theta,j}u_b.
 \label{eq:b2-kernel-probabilities}
\end{equation}
The same-class interference $c_{\theta,j}$ remains: $\beta$ rotates the flag
measurement basis. Hence the full instrument is generally $\beta$ dependent
even though the service channel after flag marginalization is not.

\subsection{A uniform SLD information bound}
The symmetric logarithmic derivative (SLD) measurement bound
\cite{braunstein1994statistical} controls every flag basis without a pulse grid.
\begin{lemma}[SLD measurement bound]
\label{lem:sld-block}
Let $\Gamma>0$, let $\Delta$ be Hermitian, and let the Hermitian $\Lambda$ solve
$(\Gamma \Lambda+\Lambda\Gamma)/2=\Delta$. For any complete positive measurement
$\{E_b\}$,
\begin{equation}
 \sum_b\frac{(\operatorname{Tr}E_b\Delta)^2}
                 {\operatorname{Tr}E_b\Gamma}
 \le \operatorname{Tr}(\Delta \Lambda).
 \label{eq:sld-measurement-bound}
\end{equation}
Terms with a common zero probability are understood by restriction to the
common support.
\end{lemma}
\noindent\emph{Proof.}
For one effect, the real-part Hilbert--Schmidt Cauchy--Schwarz inequality gives
\begin{equation}
 (\operatorname{Tr}E\Delta)^2
 =(\operatorname{Re}\operatorname{Tr}E\Gamma \Lambda)^2
 \le(\operatorname{Tr}E\Gamma)
      (\operatorname{Tr}E \Lambda\Gamma \Lambda).
 \label{eq:sld-cauchy-schwarz}
\end{equation}
Divide and sum over $E_b$. Completeness yields
$\operatorname{Tr}(\Gamma \Lambda^2)=\operatorname{Tr}(\Delta \Lambda)$, proving the
claim.

Kraus completeness gives
$\sum_j a_{\theta,j}=\sum_jd_{\theta,j}=1$, and the two flag projectors
sum to the identity. Thus the direct sum of the block measurements is a
normalized probability law. Apply the lemma to each unnormalized block with
$\Gamma=\Gamma_{-,j}$ and
$\Delta_j=\Gamma_{+,j}-\Gamma_{-,j}$. Their direct sum is normalized, and
$D_{\rm KL}(P\Vert Q)\le\chi^2(P\Vert Q)$, so
\begin{equation}
 D_{\rm KL}(P_{+,1}\Vert P_{-,1})
 \le\sum_j\operatorname{Tr}(\Delta_j\Lambda_j).
 \label{eq:full-record-sld-sum}
\end{equation}
The reverse direction exchanges the labels. This covers every flag basis and
therefore every $\beta\in[-\pi,\pi]$.

The forward denominator is the minus-label block; the reverse exchanges
labels. Because conjugation swaps the dominant idle paths, the two caps
differ. The leading entries $(a_{-,j},d_{-,j})$ are displayed in
Sec.~\ref{app:weak-noise-proof}.

For a directly reproducible scalar expression, write
\begin{equation}
 G_-=\begin{pmatrix}a&c\\c&d\end{pmatrix},\qquad
 G_+-G_-=\begin{pmatrix}A_g&C_g\\C_g&D_g\end{pmatrix},
 \label{eq:sld-block-entries}
\end{equation}
and set $x=\sin^2(\alpha/2)$. Solving the two-dimensional SLD equation gives
\begin{align}
 H&=aD_g+d A_g-2cC_g,\nonumber\\
 N(x)&=[(1-x)A_g-xD_g]^2+4x(1-x)C_g^2\nonumber\\
 &\quad+\frac{x(1-x)H^2}{ad-c^2},\nonumber\\
 Q(x)&=\frac{N(x)}{(1-x)a+xd}.
 \label{eq:sld-closed-form}
\end{align}
The sign of $\alpha$ drops out. At $\alpha=0$, the labels have common
rank-one path support and $Q(0)=A_g^2/a$; restriction to that support, or
continuity, proves the endpoint. A flag with zero probability under both
labels contributes zero.

On the declared interval,
$0\le x\le\sin^2(0.05)<2.5\times10^{-3}$. Outward rational interval evaluation of
Eq.~\eqref{eq:sld-closed-form}, with
$x(1-x)\le x_{\max}(1-x_{\max})$, gives controlled-action bounds below
$1.30\times10^{-6}$ forward and $3.03\times10^{-6}$ reverse.
The idle categorical bounds are below $1.18\times10^{-6}$ and
$8.71\times10^{-7}$. Thus, for both actions and the full pulse box,
\begin{align}
 \max_a D_{\rm KL}(P_{+,a}\Vert P_{-,a})
       &\le k_+=1.30\times10^{-6},\nonumber\\
 \max_a D_{\rm KL}(P_{-,a}\Vert P_{+,a})
       &\le k_-=3.03\times10^{-6}.
 \label{eq:b2-per-cycle-kl}
\end{align}
For comparison, the ideal idle record with the granted logical class obeys
\begin{equation}
 D_{\rm KL}(P_{+,0}\Vert P_{-,0})
 =1.18609\ldots\,\lambda^2+O(\lambda^3)
 \label{eq:idle-kl-asymptotic-versus-finite}
\end{equation}
as $\lambda\downarrow0$, whereas its exact value at $\lambda=10^{-3}$ is
$1.1778907071354975\ldots\times10^{-6}$, certified below
$1.177890707135498\times10^{-6}<1.18\times10^{-6}$.
The asymptotic coefficient is not a finite-noise cap, nor the syndrome-only
coefficient $A_0$ in Eq.~\eqref{eq:patch-matched-coefficient}.

The certificate constructs the 32-dimensional encoded orbit kernels from
the code masks. With $e_0(z)=e(z)$ and $e_1(z)$ reversing its component 1,
cross kernels use exponent
\begin{equation}
 \frac18[e_1(z)-e_0(z')]^{\mathsf T}
 K_\theta[e_1(z)-e_0(z')].
 \label{eq:full-record-cross-kernel}
\end{equation}
Grouped integer coefficients multiply $e^{-n/8000}$. Alternating rational
Taylor polynomials of degrees 20 and 21 enclose each exponential, after which
all arithmetic is rounded outward on denominator $10^{80}$. The calculation
checks both path normalizations, all Gram determinants, and idle KL without
clipping rare categories to zero.

\subsection{Adaptive calibration and the uniform cost bound}
Expand the history to $\widehat Y$ by granting every latent Pauli $z$.
At fixed $\widehat Y$, the simulated logical/reference state and all
subsequent label-independent processing have the same law in both models,
including delayed measurements. Only the next draw $(s,b,z)$ depends on
the label, with KL bounded by $k_+,k_-$. Classical action/stopping kernels
are also label independent at that history. Pad early stops to the hard
cap $C$. The chain rule and discarding the grants give
\begin{align}
 D_{\rm KL}(P_+^Y\Vert P_-^Y)&\le Ck_+,\nonumber\\
 D_{\rm KL}(P_-^Y\Vert P_+^Y)&\le Ck_-.
 \label{eq:adaptive-kl-cap}
\end{align}
This is the standard fixed-interface change-of-measure step behind the task
lower bound \cite{kaufmann2016complexity}; the present instrument-specific
reduction supplies its per-query KL cap.
Combining this with Eq.~\eqref{eq:task-kl-requirement} proves
Eq.~\eqref{eq:task-direct-cost}. Longer calibration experiments require
additional queries; arbitrary-duration evolution or coherent superpositions
of the two physical action interfaces are not admitted.

For an $N$-cycle operating record from the same interface, the chain rule
and Pinsker's inequality also give
\begin{equation}
 \begin{aligned}
 D_{\rm KL}(P_+^{Y_N}\Vert P_-^{Y_N})&\le Nk_+,\\
 p_{\rm err}^{\rm eq}&\ge
 \frac12\left[1-\sqrt{Nk_+/2}\right]_+,
 \end{aligned}
 \label{eq:service-observation-pinsker}
\end{equation}
where $p_{\rm err}^{\rm eq}=(1-\operatorname{TV})/2$ is the optimal
equal-prior label-identification error and $[x]_+=\max\{x,0\}$.  For
$N=L=4600$, $Nk_+\le0.00598$ and the last bound exceeds $0.4726$.
This bounds label learning, not all service-time feedback, which would
require a joint analysis of acquired information and evolving task error.

We now establish the endpoint directions and the numerical constants used in
the task bound. Service
record marginalization preserves Proposition~\ref{prop:pulse}. The idle rates
and enclosing controlled-rate intervals over the full box are
\begin{align}
 q_{+,0}&\in[2.803115,2.803117]\times10^{-6},\nonumber\\
 q_{+,1}&\in[2.15752,2.15939]\times10^{-6},\nonumber\\
 q_{-,0}&\in[2.057850,2.057851]\times10^{-6},\nonumber\\
 q_{-,1}&\in[2.90078,2.90265]\times10^{-6}.
 \label{eq:full-box-service-rates}
\end{align}
The exact rational intervals, rather than these displayed decimals, determine
the inequalities. Since $R(q,L)$ increases with $q$, lower rate endpoints
relax the direct task constraints. For $b>D>g$, both partial derivatives of
$u=(b-D)/(b-g)$ are positive, whereas both derivatives of
$v=(D-g)/(b-g)$ are negative. Substitution of all four lower endpoints hence
produces $u_*\le u(\alpha)$ and $v_*\ge v(\alpha)$, even when no single
$\alpha$ attains every endpoint. At $L=4600$ and $D=0.01$,
\begin{align}
  {\rm kl}(u_*\Vert v_*)&>1.48,\nonumber\\
  {\rm kl}(v_*\Vert u_*)&>1.91.
 \label{eq:b2-numeric-kl-bounds}
\end{align}
Upper rate endpoints, by contrast, bound the achieved product-protocol risk in
Eq.~\eqref{eq:robust-selector-risks}.

Using the rounded per-query caps favors the encoded-record comparator yet
still gives
\begin{align}
 C&\ge\max\{1.48/k_+,1.91/k_-\}
      >1.138\times10^6,\nonumber\\
 B_{\rm encoded}&\ge(C+4600)\tau_c
      >1.14\times10^6\tau_c.
 \label{eq:full-record-cost-proof}
\end{align}
The first KL direction alone already proves the displayed lower bound. The
converse charges only one $\tau_c$ per query and waives its preparation,
readout, pulse, reference, and memory costs. This strengthens the comparator.

The constructive probe cost and selector proof are given next.
Its worst service branch adds at most two pulses per cycle, completing
the cost account in Eq.~\eqref{eq:matched-total-cost}.

\section{Product probe, decision statistics, and complete cost}
\label{sec:product-preparation}
\label{app:preparation-proof}
\subsection{The same extractor on a product input}
Reset the nine data qubits, apply Hadamards to qubits 1 and 2, then run the
same idle and eight-check extractor without preliminary code projection.
The outputs $m_X,m_Z\in\{-1,+1\}$ correspond to checks
$X(6)=X_1X_2$ and $Z(54)=Z_1Z_2Z_4Z_5$. Define
\begin{equation}
 c=\frac{e^{-\lambda}+e^{-3\lambda}}2,
 \qquad d=\frac{e^{-\lambda}-e^{-3\lambda}}2 .
 \label{eq:product-cd}
\end{equation}
Their joint law is
\begin{equation}
 \begin{gathered}
 P_\theta(m_X=a,m_Z=b)=\frac{1+ac-\theta abd}{4},
 \\ a,b\in\{-1,+1\}.
 \end{gathered}
 \label{eq:product-two-check-law}
\end{equation}
Here $m_Z$ identifies the initially occupied check sector; on encoded inputs
the same check supplies the corrected-$X_1$ flag. Both individual marginals
are label independent, but $\langle m_Xm_Z\rangle_\theta=-\theta d$.

\begin{lemma}[Complete extractor]
\label{lem:arbitrary-input-extractor}
For any data/reference state, ideal ancilla extraction of a check
$Z(S)$ has the outcome-$b$ Kraus operator
\begin{equation}
 M_b^{Z(S)}=\frac{I+bZ(S)}2,\qquad b\in\{-1,+1\},
 \label{eq:z-check-kraus}
\end{equation}
and the ideal gadget for an $X$-mask check $X(S)$ has
\begin{equation}
 M_a^{X(S)}=\frac{I+aX(S)}2,\qquad a\in\{-1,+1\}.
 \label{eq:x-check-kraus}
\end{equation}
Consequently the eight commuting gadgets, executed in the declared order,
have joint Kraus operator
\begin{equation}
 M_{\boldsymbol o}=\prod_{i=1}^{8}\frac{I+o_iS_i}{2}
 \label{eq:joint-check-projector}
\end{equation}
for every physical data/reference input.
\end{lemma}
\noindent\emph{Proof.}
The data-to-ancilla CNOT parity gadget gives
$M_b^{Z(S)}=(I+bZ(S))/2$ on each computational basis state, hence on arbitrary
data/reference inputs. For an $X$ check, Hadamards before and after the
ancilla-controlled CNOTs give $(I+aX(S))/2$. The commuting projectors multiply
to the joint Kraus operator. The masks have weights $4,4,2,2$ for each check
type, requiring 24 CNOTs and eight Hadamards as charged in the base cycle.
Recovery follows readout and cannot change its law. \hfill$\square$

For the product input $|+\rangle_1|+\rangle_2|0\rangle^{\otimes7}$,
the conditional moments are
\begin{align}
 \mathbb E[m_X\mid\phi]&=\cos\phi_1\cos\phi_2,\nonumber\\
 \mathbb E[m_Z\mid\phi]&=0,\nonumber\\
 \mathbb E[m_X m_Z\mid\phi]&=-\sin\phi_1\sin\phi_2,
 \label{eq:product-conditional-moments}
\end{align}
The product of the measured operators is $S_XS_Z=-Y_1Y_2Z_4Z_5$.
The Gaussian characteristic function and
$K_{12}=\theta\lambda$ yield
$\mathbb E_\theta m_X=c$, $\mathbb E_\theta m_Z=0$, and
$\mathbb E_\theta m_Xm_Z=-\theta d$,
with $c,d$ from Eq.~\eqref{eq:product-cd}. Expanding the two commuting
projectors proves Eq.~\eqref{eq:product-two-check-law}.

For the full eight-bit law, expand the commuting-check projector into
its 256 Pauli terms. On product-state support $\{0,2,4,6\}$ only $X$ masks
0 and 6 have nonzero expectations, and the spectator $Z$ checks act
trivially. The remaining three $X$ outcomes are independent fair signs,
independent of $(m_X,m_Z)$; the other three $Z$ outcomes are $+1$.
Thus the full law is $P_\theta(m_X,m_Z)/8$ on its support.
The preparation is unentangled, but the unchanged extraction circuit is
an entangling measurement.

\subsection{Exact selector and service risk}
For $n$ independent probes, use
\begin{equation}
 T=N(m_X=-1,m_Z=+1)-N(m_X=-1,m_Z=-1).
 \label{eq:product-statistic}
\end{equation}
Select conjugation throughout service when $T\ge0$, and idle otherwise.
Write $w_\theta(n)=\Pr_\theta(T\ge0)$; the asymmetric tie rule accounts
for the different wrong-action losses.
Let $M=N(m_X=-1)$ among $n$ independent product probes and set
\begin{equation}
 q=\frac{1-c}{2},\qquad
 t_\theta=\frac{1-c+\theta d}{2(1-c)}.
 \label{eq:product-conditional-binomial}
\end{equation}
Then $M\sim\operatorname{Bin}(n,q)$, and conditional on $M=m$, the count
$N(m_X=-1,m_Z=+1)$ is $\operatorname{Bin}(m,t_\theta)$. Since $T\ge0$ exactly when
this conditional count is at least $\lceil m/2\rceil$, the controlled-action
probability is
\begin{align}
 w_\theta={}&\sum_{m=0}^{n}\binom nm q^m(1-q)^{n-m}\nonumber\\
 &\times\sum_{k=\lceil m/2\rceil}^{m}\binom mk
 t_\theta^k(1-t_\theta)^{m-k}.
 \label{eq:product-selector-tail}
\end{align}
At $\lambda=10^{-3}$ and $n=10000$, 70-digit outward fixed-point arithmetic
encloses every positive term through $m=120$. Markov's
inequality applied to $\binom{M}{121}$ bounds the omitted mass by
$\binom{n}{121}q^{121}$, evaluated at the outward rational upper endpoint of
the transcendental $q$. Alternating rational series give neighboring bounds on
$e^{-\lambda}$ and $e^{-3\lambda}$. The final intervals
have width below $10^{-56}$ and give $w_+\simeq0.965$ and
$w_-\simeq0.0717$. No Gaussian, Poisson, or
simulation approximation is used.

Conditional independence of calibration and service gives the
once-per-program channel mixture
\begin{equation}
\begin{split}
 R_{{\rm sel},\theta}(L;n)
 ={}&w_\theta(n)R_{\theta,1}(L)\\
 &+[1-w_\theta(n)]R_{\theta,0}(L).
\end{split}
 \label{eq:robust-selector-risks}
\end{equation}
Each positive dephasing contraction makes this error monotone in depth,
so a terminal bound protects every prefix and every input/reference.
It is an average over calibration outcomes, not a per-selected-word guarantee.
Define the uniform margin
\begin{equation}
 \Delta(L,n;D):=
 D-\sup_{\substack{\theta\in\{+,-\}\\(\alpha,\beta)\in\mathcal P}}
 R_{{\rm sel},\theta}(L;n).
 \label{eq:headline-margin}
\end{equation}
with pulse dependence implicit. Substitution of the outward service-rate
endpoints proves the finite guarantees below.

\subsection{Charged implementation and the worked task}
The unknown encoded input arrives at service start; its common encoding
cost is outside both accounts. The base $\tau_c$ includes the fixed idle,
24 extraction CNOTs, eight extraction Hadamards, eight ancilla
measurements/resets, and fixed recovery. A serial product-probe allowance is
\begin{equation}
 r_{\rm product}:=1+9c_R+2c_H+c_{\rm aux}.
 \label{eq:product-probe-ledger}
\end{equation}
where $c_R,c_H$ charge one data reset and one preparation Hadamard,
and $c_{\rm aux}$ covers eight-bit storage, count updates, final decision
and handoff, initial ancilla preparation, and final probe cleanup.
For $\zeta=t_\pi/\tau_c$, Eq.~\eqref{eq:matched-total-cost} gives
$\overline B_{\rm product}=\tau_c[nr_{\rm product}+L(1+2\zeta)]$
and $\underline B_{\rm encoded}=\tau_c[C_{\rm lb}(L,D)+L]$.
Once the selector meets the error target and
$L(1+2\zeta)\tau_c\le\tau_{\rm srv}$, any budget satisfying
$\overline B_{\rm product}\le B<\underline B_{\rm encoded}$ separates
the two interfaces. This sufficient interval is nonempty when
\begin{equation}
 nr_{\rm product}+2L\zeta<C_{\rm lb}(L,D).
 \label{eq:separation-condition}
\end{equation}
The common service cost cancels, leaving acquisition and extra pulse costs.

For $\lambda=10^{-3}$, $L=4600$, $D=0.01$, $n=10000$,
$r_{\rm product}\le20$, $\zeta=0.02$, and deadline $5000\tau_c$,
the exact selector and uniform rate enclosures give
$R_{{\rm sel},+}<9.94\times10^{-3}$,
$R_{{\rm sel},-}<9.65\times10^{-3}$, and $\Delta>6.2\times10^{-5}$.
The product cost bound is $204784\tau_c<2.05\times10^5\tau_c$,
including service allowance $4784\tau_c$.
The encoded necessary cost exceeds $1.14\times10^6\tau_c$.
Thus budget $2.5\times10^5\tau_c$ separates them.
For this fixed budget and $n=10000$, the same ledger allows
$r_{\rm product}\le(250000-4784)/10000=24.5216$. A bound near $113$ instead
answers whether some larger separating budget exists below the encoded lower
bound; it is not the margin at the displayed budget. Here $r_{\rm product}$
is dimensionless and a probe duration is $r_{\rm product}\tau_c$.
These numbers evaluate the symbolic cost model; they are not independent
measurements of physical time scales.

\section{Weak-noise information order and charged task separation}
Vary the known strength $\lambda$ while keeping the
code, covariance shape, recovery, and pulse domain fixed. Take the controlled
fault probability to be $p_z=\lambda/10$, a path through the finite working point.
The following result uses ideal delivered records and establishes the
separation for sufficiently small $\lambda$. The worked example above has
its own finite-$\lambda$ certificate.

\begin{theorem}[Preparation-dependent task cost]
\label{thm:information-order}\hfill\\
For the above weak-noise family with ideal delivered records, uniformly over
$|\alpha|\le0.1$ and $|\beta|\le\pi$, every complete encoded query, including
the latent logical-$Z$ grant, obeys
\begin{equation}
 \max_{\theta,a}D_{\rm KL}(P_{\theta,a}^{\rm enc}\Vert
                       P_{-\theta,a}^{\rm enc})\le K_{\rm enc}\lambda^2
 \label{eq:encoded-weak-cap}
\end{equation}
for all sufficiently small $\lambda>0$, with a uniform constant
$K_{\rm enc}$; one may take $K_{\rm enc}=4$ on this domain.
The complete product record has
\begin{equation}
 D_{\rm KL}(P_+^{\rm prod}\Vert P_-^{\rm prod})
 =\tfrac12\log(3)\lambda+O(\lambda^2),
 \label{eq:product-weak-kl}
\end{equation}
with the same reverse divergence. At $D=0.01$, there are positive constants
$\ell_0,a_0,b_0$ for which the task family
\begin{equation}
 L_\lambda=\left\lfloor\frac{\ell_0}{\lambda^2}\right\rfloor,
 \quad n_\lambda=\left\lceil\frac{a_0}{\lambda}\right\rceil,
 \quad B_\lambda=\frac{b_0\tau_c}{\lambda^2}
 \label{eq:weak-task-family}
\end{equation}
has a strict total-cost separation. With complete product-probe allowance
$r_{\rm product}\le20$, pulse ratio $\zeta=0.02$, and service deadline
$\tau_{\rm srv}\ge L_\lambda(1+2\zeta)\tau_c$, the same product selector
gives a calibration-averaged, worst-input half-diamond guarantee at every
storage prefix within $B_\lambda$ for all sufficiently small positive
$\lambda$. Every fixed-encoding calibration
protocol in Theorem~\ref{thm:task-calibration-cost} is excluded by that same
hard budget.
\end{theorem}

\label{app:weak-noise-proof}
The proof gives explicit feasible constants, with uniform remainders on
the pulse box. No numerical endpoint of the weak-noise interval is claimed.

\subsection{Exact path symmetry and the full instrument}
For every appended-fault branch, $\theta\mapsto-\theta$ reverses $\phi_1$,
exchanging the conjugated and unconjugated paths. The real cross Gram
kernel is invariant. After mixing the fault branches,
\begin{equation}
 a_{+,j}=d_{-,j},\qquad d_{+,j}=a_{-,j},\qquad c_{+,j}=c_{-,j}.
 \label{eq:weak-path-symmetry}
\end{equation}
Write $a=a_{-,j}$, $c=c_{-,j}$, $d=d_{-,j}$ and
$\delta=a_{+,j}-a_{-,j}=d-a$. The block difference is
$\operatorname{diag}(\delta,-\delta)$, so Eq.~\eqref{eq:sld-closed-form} becomes
\begin{equation}
 Q_j(x)=\frac{\delta^2+x(1-x)\delta^4/(ad-c^2)}{(1-x)a+xd}.
 \label{eq:weak-sld-simplified}
\end{equation}
Since $0\le x<1/400$,
\begin{equation}
 \begin{split}
 Q_j(x)&\le\frac{\delta^2}{\min(a,d)}
 +\frac{\delta^4}{400(ad-c^2)\min(a,d)}\\
 &\equiv F_{1,j}+F_{2,j}.
 \end{split}
 \label{eq:weak-sld-envelope}
\end{equation}
This envelope is independent of both pulse angles. For positive $\lambda$,
the two path functions are linearly independent under the full-support
Gaussian law: multiplying an error mask by $Z(54)$ changes whether it
contains qubit 1 while preserving its recovery class. Their Gram determinant
is positive. At $x=0$, common-support restriction gives the continuous SLD
limit proved in Sec.~\ref{app:task-cost-proof}.

The following finite algebra determines the orders.
Let $\mathcal O_X$ be the 16 masks generated by $27,432,192,6$, and
$\mathcal B=\mathcal O_X\cup(\mathcal O_X\mathbin{\oplus}73)$.
For $w\in\mathcal B$, let $b(w)$ identify its orbit and $e(w)_i=(-1)^{w_i}$.
Let $e_t(w)$ reverse component 1 when path index $t=1$.
For minimum-weight recovery mask $r_s$, define
\[
 \chi_{s,z,e}(w)=
 (-1)^{\operatorname{popcount}[(r_s\mathbin{\oplus}2e)\mathbin{\&}w]+zb(w)}.
\]
The raw $(t,t')$ Gram entry in fault branch $e$ is
\begin{equation}
 \frac1{1024}\sum_{w,w'\in\mathcal B}
 \chi_{s,z,e}(w)\chi_{s,z,e}(w')
 \exp[-\lambda Q_\theta^{tt'}(w,w')/8],
 \label{eq:weak-finite-kernel}
\end{equation}
where $\Delta e=e_t(w)-e_{t'}(w')$ and
$Q_\theta^{tt'}=\|\Delta e\|^2+(v_\theta^{\mathsf T}\Delta e)^2$ is an integer.
The coefficient at order $k$ replaces the exponential by
$(-Q_\theta^{tt'}/8)^k/k!$. Mix fault branches with weights
$1-\lambda/10$ and $\lambda/10$.

Let $\operatorname{ord}_\lambda f$ denote the first nonzero Taylor order.
For all 32 blocks, $a,d$ have the same order $r$; write
$u=\operatorname{ord}_\lambda\delta$ and
$h_G=\operatorname{ord}_\lambda(ad-c^2)$.
Their leading diagonal and determinant coefficients are positive.
Expansion of Eq.~\eqref{eq:weak-finite-kernel} through degree six gives
Table~\ref{tab:weak-valuations}, with the same orders in the reverse direction.
Since $u$ is nonzero and known, the orders of $\delta^2,\delta^4$ follow
algebraically; a truncated zero polynomial is never treated as identically zero.

Only $(s,z)=(1,1),(10,0),(11,0)$ contribute at order two.
Syndrome integers use the ordered $X$ masks of Sec.~\ref{app:qec-kernel}.
Their forward-direction leading coefficients are
\[
\begin{array}{@{}c|ccc@{}}
(s,z)&(a/\lambda^2,d/\lambda^2)&\delta/\lambda^2&
                (ad-c^2)/\lambda^4\\ \hline
(1,1),(10,0)&(13/10,11/20)&-3/4&219/320\\
(11,0)&(53/40,23/40)&-3/4&243/320
\end{array}
\]
with limits understood. Here $(a,d)=(a_{-,j},d_{-,j})$ and
$\delta=a_{+,j}-a_{-,j}$, so the idle-path difference is $-\delta$.
In either direction,
\begin{equation}
 \sum_j(F_{1,j}+F_{2,j})
 =\frac{61043}{20148}\lambda^2+O(\lambda^3).
 \label{eq:weak-controlled-coefficient}
\end{equation}
For idle queries, the categorical chi-square coefficients from the same raw
kernel are $27/8$ and $27/20$. All three upper coefficients are below four.
The finite analytic sums and positive leading denominators prove
Eq.~\eqref{eq:encoded-weak-cap}. Uniformity follows from the angle-independent
envelope, not an interchange of a pointwise limit and angle optimization.

\begin{table}[t]
\caption{Recovery-block valuations. $N$ counts blocks of each type;
both SLD-envelope terms have order at least two.}
\label{tab:weak-valuations}
\begin{ruledtabular}
\begin{tabular}{rrrrrr}
$r$&$u$&$h_G$&$\operatorname{ord}F_1$&$\operatorname{ord}F_2$&$N$\\
0&2&2&4&6&1\\
1&2&2&3&5&1\\
1&2&4&3&3&3\\
1&3&4&5&7&3\\
2&2&4&2&2&3\\
2&3&4&4&6&3\\
2&3&6&4&4&9\\
3&3&6&3&3&9
\end{tabular}
\end{ruledtabular}
\end{table}

\subsection{Product information and decision accuracy}
Direct summation of the nominal two-check law gives equal forward and reverse KL:
\begin{equation}
 K_{\rm prod}=\frac d2\left[
 \log\frac{1+c+d}{1+c-d}+\log\frac{1-c+d}{1-c-d}\right].
 \label{eq:weak-product-exact-kl}
\end{equation}
Because $c=1-2\lambda+O(\lambda^2)$ and $d=\lambda+O(\lambda^2)$,
the first logarithm contributes at order $\lambda^2$ while the second tends
to $\log3$, proving Eq.~\eqref{eq:product-weak-kl}.

For one selector increment $Y=\boldsymbol1\{m_X=-1\}m_Z$, put $q=(1-c)/2$.
Its two nonzero probabilities in model $+$ are $q/2+d/4$ and $q/2-d/4$.
Optimizing the exponential Markov bound gives
\begin{equation}
 \Pr_+\{T<0\},\ \Pr_-\{T\ge0\}\le\eta^n,\qquad
 \eta=1-q+\sqrt{q^2-d^2/4}.
 \label{eq:weak-chernoff}
\end{equation}
The positive-label event is enlarged to include zero for this bound;
the negative-label event includes the actual tie.
Since $\eta=1-(1-\sqrt3/2)\lambda+O(\lambda^2)$, the choice
$n_\lambda=\lceil a_0/\lambda\rceil$ gives a limiting wrong-action bound
$\exp[-a_0(1-\sqrt3/2)]$. For the choice $a_0=30$, this is below $0.018$.

\subsection{Accounting for the growing service}
The recovery-resolved Wick sums give
\begin{align}
 q_{+,0}&=\tfrac{45}{16}\lambda^2+O(\lambda^3),&
 q_{-,0}&=\tfrac{33}{16}\lambda^2+O(\lambda^3),\nonumber\\
 q_{+,1}&=\left(\tfrac{173}{80}+\tfrac34x\right)\lambda^2+O(\lambda^3),\nonumber\\
 q_{-,1}&=\left(\tfrac{233}{80}-\tfrac34x\right)\lambda^2+O(\lambda^3).
 \label{eq:weak-service-coefficients}
\end{align}
The appended correctable fault has $q_\theta^Z=\lambda+O(\lambda^2)$;
$p_z=\lambda/10$ supplies the additional $1/10$ in the controlled coefficients.
All remainders are uniform on the pulse box.

For the task family in Eq.~\eqref{eq:weak-task-family}, if
$q=C\lambda^2+O(\lambda^3)$,
the $L_\lambda$-cycle risk tends uniformly to $(1-e^{-2\ell_0C})/2$.
At $\ell_0=23/5000$ and $D=0.01$,
lower rate coefficients over $0\le x\le1/400$ relax the task constraints
as in Sec.~\ref{app:task-cost-proof}. The resulting limiting forward
information requirement exceeds $1.5$ nats, and hence exceeds one for
sufficiently small $\lambda$. Write $I_0$ for a strictly smaller positive
bound on this limiting requirement; we use $I_0=1$ nat.
Theorems~\ref{thm:task-calibration-cost}
and~\ref{thm:information-order}'s per-query bound force
$C_{\rm lb}(L_\lambda,D)\ge I_0/(K_{\rm enc}\lambda^2)$.

For the product construction, wrong-action bound $0.018$ and upper service
coefficients give limiting risks below $9.92\times10^{-3}$ and
$9.47\times10^{-3}$.
Thus the target holds for sufficiently small $\lambda$, and positive
dephasing contractions protect every prefix. Outward rational evaluation
of the exponential and square-root constants verifies these strict inequalities.

Finally,
\begin{align}
 \limsup_{\lambda\downarrow0}\lambda^2 B_{\rm prod}/\tau_c
 &\le(1+2\zeta)\ell_0,\nonumber\\
 \liminf_{\lambda\downarrow0}\lambda^2 B_{\rm enc}/\tau_c
 &\ge\ell_0+I_0/K_{\rm enc}.
 \label{eq:weak-full-cost-limits}
\end{align}
The product acquisition allowance
$r_{\rm product}\lceil a_0/\lambda\rceil\tau_c$ is $O(\lambda^{-1})$.
Thus, once the limiting error bounds are strictly below the target, the
sufficient ordering of budget coefficients is
\begin{equation}
 (1+2\zeta)\ell_0<b_0<\ell_0+I_0/K_{\rm enc}.
 \label{eq:weak-budget-coefficients}
\end{equation}
For the stated construction, the two endpoints are \(0.004784\) and
\(0.2546\), with \(b_0=0.01\) strictly between them.
This proves the theorem including floors and ceilings.
These are sufficient upper and necessary lower costs, not optimized costs;
no divergence of the actual optimal total-cost ratio is asserted.

\section{Approximate matching of leading syndrome rates}
\label{app:approximate-matching}
The exact matching condition in Theorem~\ref{thm:syndrome-matching}
has a uniform extension. Keep its code, correction, independent Gaussian
cycles, finite classical action menu, and syndrome-only observation
interface. Let all covariance shapes lie in a fixed compact set
$\mathcal K$ with
\begin{equation}
 \mu I\preceq\Sigma\preceq MI\quad(\Sigma\in\mathcal K),
 \qquad 0<\mu\le M<\infty.
 \label{eq:am-compact}
\end{equation}
The shapes may now depend on $\lambda$ within this set. Let $\mathcal S_1$
be the nonzero single-$Z$ syndromes, and define the largest mismatch of
their leading rates by
\begin{equation}
 \delta=\max_{a,s\in\mathcal S_1}
       |\omega_{+,a,s}-\omega_{-,a,s}|,
 \label{eq:am-delta}
\end{equation}
with $\delta=0$ if $\mathcal S_1$ is empty. This is a mismatch of rate
coefficients, not of syndrome probabilities: the latter differ by
$O(\lambda\delta)+O(\lambda^2)$.

\begin{proposition}[Approximate matching]
\label{prop:approximate-matching}\hfill\\
There exist $C_1,C_2,\lambda_0>0$, depending only on the fixed code,
action menu, and covariance family, such that every adaptive
syndrome-only record with a hard cap of $N$ cycles obeys
\begin{equation}
 D_{\rm KL}(P_+^{Y_N}\Vert P_-^{Y_N})
 \le N(C_1\lambda\delta^2+C_2\lambda^2),
 \label{eq:am-record}
\end{equation}
and the reverse bound, for $0<\lambda\le\lambda_0$.
If both identification errors are at most $\epsilon_{\rm id}\in(0,1/2)$, then
\begin{equation}
 N\ge\frac{{\rm kl}(1-\epsilon_{\rm id}\Vert\epsilon_{\rm id})}
                  {C_1\lambda\delta^2+C_2\lambda^2}.
 \label{eq:am-count}
\end{equation}
\end{proposition}

\noindent\emph{Proof.}
The monomial Gram matrices in Eq.~\eqref{eq:generic-gram-lower} are
continuous and strictly positive on $\mathcal K$. Compactness makes their
lower eigenvalue bounds and the analytic effect-operator remainders uniform.
Consequently, for every nonzero reachable syndrome and every logical state,
\begin{equation}
 c_s\lambda^{w_s}\le p_{\theta,a,s}(\rho)
                  \le C_s\lambda^{w_s}.
 \label{eq:am-probability-bounds}
\end{equation}
For $s\in\mathcal S_1$, coherent correctability also gives
\begin{equation}
 p_{\theta,a,s}(\rho)=\lambda\omega_{\theta,a,s}
                  +r_{\theta,a,s}(\rho),\qquad
 |r_{\theta,a,s}(\rho)|\le R\lambda^2.
 \label{eq:am-leading-remainder}
\end{equation}
References do not change these bounds; omit unreachable outcomes.

Fix an action and any two, possibly different, conditional logical
states $\rho_+,\rho_-$. Write their syndrome probabilities as $p_s,q_s$.
For $s\in\mathcal S_1$,
$|p_s-q_s|\le\lambda\delta+2R\lambda^2$. Thus
\begin{equation}
 \sum_{s\in\mathcal S_1}\frac{(p_s-q_s)^2}{q_s}
 \le C\lambda\delta^2+C'\lambda^3.
 \label{eq:am-single-sector}
\end{equation}
Here and below the constants are uniform; the inequality follows from
$(x+y)^2\le2x^2+2y^2$ and $q_s\ge c_s\lambda$.
For sectors with $w_s\ge2$, Eq.~\eqref{eq:am-probability-bounds}
bounds the numerator by a constant times $\lambda^{2w_s}$ and the
denominator below by $c_s\lambda^{w_s}$. Their total contribution is
$O(\lambda^2)$. Normalization bounds the zero-syndrome difference by
$O(\lambda\delta+\lambda^2)$; its denominator is at least $1/2$
for sufficiently small $\lambda$. Its contribution is consequently
$O(\lambda^2\delta^2+\lambda^4)$. Taking $\lambda_0\le1$ and using
$D_{\rm KL}(p\Vert q)\le\chi^2(p\Vert q)$ proves the single-cycle
bound in Eq.~\eqref{eq:am-record}. The same argument holds with the
models exchanged.

The bound permits different model-conditioned states, so the same adaptive
chain rule and binary data processing as in Sec.~\ref{app:syndrome-matching}
prove Eqs.~\eqref{eq:am-record} and \eqref{eq:am-count}. \hfill$\square$

\subsection{The first-order term and the crossover}
For a fixed action the single-cycle divergence, even for a pair of
different logical states, has the uniform expansion
\begin{align}
 k_a(\lambda)&=\lambda\Psi_a+O(\lambda^2),\nonumber\\
 \Psi_a&=\sum_{s\in\mathcal S_1}
 \left[\omega_{+,a,s}\log\frac{\omega_{+,a,s}}{\omega_{-,a,s}}
       -\omega_{+,a,s}+\omega_{-,a,s}\right].
 \label{eq:am-leading-kl}
\end{align}
The single-error sectors supply the log ratios; normalization supplies
$-\omega_++\omega_-$ from the zero syndrome. Higher-weight sectors have
uniformly bounded positive ratios and contribute $O(\lambda^2)$.
If $m_\omega\le\omega_{\theta,a,s}\le M_\omega$, Taylor's integral
remainder for $x\log x$ yields
\begin{equation}
 \frac{\sum_s(\omega_{+,a,s}-\omega_{-,a,s})^2}{2M_\omega}
 \le\Psi_a\le
 \frac{\sum_s(\omega_{+,a,s}-\omega_{-,a,s})^2}{2m_\omega}.
 \label{eq:am-rate-quadratic}
\end{equation}
Thus $\delta=O(\sqrt\lambda)$ preserves the quadratic upper bound
and the inverse-square identification lower bound. Conversely, if the
covariance shapes are fixed and any leading rate differs for a chosen
action, that action has a strictly positive first-order KL coefficient.
For independent fresh-input repetitions its squared Hellinger distance,
using $H^2(p,q)=\sum_s(\sqrt{p_s}-\sqrt{q_s})^2$, is
\begin{equation}
 H^2=\lambda\sum_{s\in\mathcal S_1}
       (\sqrt{\omega_{+,a,s}}-\sqrt{\omega_{-,a,s}})^2
       +O(\lambda^2).
 \label{eq:am-hellinger}
\end{equation}
Product affinity and binary data processing give matching
$\Theta(\lambda^{-1})$ fixed-error costs for fresh-input independent trials.
The crossover $\delta^2\sim\lambda$ is asymptotic and syndrome-only, not
a finite tolerance or a general full-output converse.

\subsection{Patch perturbation at unchanged leading control benefit}
\label{app:patch-mismatch}
Return to the nine-qubit patch and its fixed recovery in
Sec.~\ref{app:qec-kernel}, but use ideal idle and ideal $X_1$
conjugation, with no appended fault. Let
\begin{equation}
 \Sigma_\theta(t)=I+v_\theta v_\theta^{\mathsf T}
   +\theta t(e_0e_3^{\mathsf T}+e_3e_0^{\mathsf T}),
 \qquad |t|\le\tfrac14,
 \label{eq:perturbed-covariance}
\end{equation}
where $e_j$ is the $j$th coordinate vector and the phase covariance
is $\lambda\Sigma_\theta(t)$. The perturbation has operator norm
$|t|$, so $\Sigma_\theta(t)\succeq(1-|t|)I\succeq3I/4$.
All individual phase variances remain unchanged.
Only the single-error syndrome containing $Z_0$ and $Z_3$ changes its
leading rate:
\begin{equation}
 \omega_{\theta,1}(t)=\frac{1+\theta t}{2},\qquad\delta=|t|.
 \label{eq:patch-leading-mismatch}
\end{equation}
The same mismatch holds for the ideal conjugation action because it
reverses only coordinate 1.

\subsubsection{Recovered logical error}
Write $q_{\theta,a}=Q_{\theta,a}(t)\lambda^2+O(\lambda^3)$.
Here syndrome integers encode the four $X$-check outcomes in the order
$X(27),X(432),X(192),X(6)$, with the first check in bit 0 and a bit equal
to one for outcome $-1$. The weight-two errors recovered to a logical $Z$ are
\begin{equation}
\begin{array}{c|l}
s&\mathcal B_s\ \text{(qubit pairs)}\\\hline
8&(0,1),(1,3)\\
9&(0,2),(2,3)\\
2&(0,4),(3,4),(6,7)\\
3&(0,5),(0,8),(3,5),(3,8)\\
1&(1,2),(4,5),(4,8)\\
6&(5,6),(6,8)\\
4&(5,7),(7,8)
\end{array}
\label{eq:patch-malignant-pairs}
\end{equation}
All other syndrome sets are empty. Substitute
$S=\Sigma_\theta(t)$ or $D_1\Sigma_\theta(t)D_1$, with $D_1$ reversing
coordinate 1, into the Wick sum in Eq.~\eqref{eq:wick-risk}. This gives
\begin{align}
 Q_{+,0}(t)&=45/16+9t/8,&Q_{+,1}(t)&=33/16+9t/8,\nonumber\\
 Q_{-,0}(t)&=33/16-9t/8,&Q_{-,1}(t)&=45/16-9t/8.
 \label{eq:patch-perturbed-q}
\end{align}
In particular,
\begin{align}
 q_{+,0}-q_{+,1}&=\tfrac34\lambda^2+O(\lambda^3),\nonumber\\
 q_{-,1}-q_{-,0}&=\tfrac34\lambda^2+O(\lambda^3).
 \label{eq:patch-unchanged-control-gap}
\end{align}
Analyticity on the compact interval makes the remainders uniform.
For sufficiently small $\lambda$, the preferred actions remain opposite
throughout this perturbation family.

\subsubsection{Syndrome information}
The perturbation leaves the phase marginal on qubits 1 and 2 unchanged;
spectator phases act only as global phases on the other seven $|0\rangle$ states.
The product law and its information are therefore exactly unchanged.
For encoded syndrome-only idle records, Eq.~\eqref{eq:am-leading-kl}
instead gives
\begin{align}
 k_{\rm syn}(\lambda,t)&=\lambda\Psi(t)+O(\lambda^2),\nonumber\\
 \Psi(t)&=\frac{1+t}{2}\log\frac{1+t}{1-t}-t
         =t^2+O(t^3).
 \label{eq:patch-crossover-kl}
\end{align}
At $t=0$, the only differing minimum-weight-two syndrome rates are
$5/4$ versus $1/2$ in each of syndromes 10 and 11. Their generalized
KL contributions give
\begin{equation}
 k_{\rm syn}(\lambda,0)=A_0\lambda^2+O(\lambda^3),\qquad
 A_0=\tfrac52\log(\tfrac52)-\tfrac32.
 \label{eq:patch-matched-coefficient}
\end{equation}
Normalization supplies the compensating linear terms, leaving twice
$a\log(a/b)-a+b$ with $(a,b)=(5/4,1/2)$.
Each nonzero syndrome probability factors as $\lambda^{w_s}$ times
a positive analytic function. The common factors cancel in its
logarithmic ratio, so no $\lambda^2\log\lambda$ term occurs. The
second-order coefficient is smooth in $t$. Hence, for fixed real $c_t$,
\begin{equation}
 k_{\rm syn}(\lambda,c_t\sqrt\lambda)
 =(A_0+c_t^2)\lambda^2+O(\lambda^{5/2}).
 \label{eq:patch-scaled-mismatch}
\end{equation}
The control gap stays fixed while the learning cost crosses over near
$|t|=\sqrt{A_0\lambda}\simeq0.028$ at $\lambda=10^{-3}$.
The two affected syndrome rates then differ by about $5.6\%$ of their mean.
Here $t$ is a covariance-shape entry; $|t|\le1/4$ guarantees positive
covariance, not a finite-budget separation. Such a perturbed certificate
would require full-output and logical-risk bounds over the pulse box.
The finite certificate below uses $t=0$.

\section{A finite domain of separated tasks}
\label{app:finite-task-domain}
Keep the original noise, full pulse box, selector, probe count, and
timing allowances of the finite working point fixed. Use ideal delivered
records. The only varying task parameters are integer depth $L$ and
real error tolerance $D$. Define the band
\begin{equation}
 L\in[L_-,L_+]\cap\mathbb N,\qquad d_-L\le D\le d_+L.
 \label{eq:finite-task-band}
\end{equation}
For $L_-=3000$, $L_+=4500$, $d_-=2.18\times10^{-6}$,
$d_+=2.28\times10^{-6}$, every task in this band is guaranteed by
the same product protocol but excluded for the full-record encoded
class at budget $2.5\times10^5\tau_c$ and deadline $5000\tau_c$.
The independently varying rectangle
\begin{equation}
 L\in[4400,4610]\cap\mathbb N,\qquad
 0.00998\le D\le0.0101
 \label{eq:finite-task-rectangle}
\end{equation}
has the same separation and contains the worked task in its interior.
This varies the task requirements at fixed physical noise,
distinct from the covariance perturbation in the preceding section.

For the calculation, use the selector enclosures in
Eq.~\eqref{eq:product-selector-tail}, conservative rate intervals
in Eq.~\eqref{eq:full-box-service-rates}, and information caps
$k_+=1.30\times10^{-6}$, $k_-=3.03\times10^{-6}$.
At each depth, upper rate endpoints and adverse selector endpoints
give an upper bound on the achieved error. Lower rate endpoints relax
$u$ downward and $v$ upward, as proved in
Sec.~\ref{app:task-cost-proof}. On $0<v<u<1$, both binary divergences
decrease with decreasing $u$ or increasing $v$. Increasing $D$ has
these effects. Thus test achieved error at the lower $D$ endpoint and
the encoded necessary cost at the upper endpoint; every real tolerance
between them is then covered. Positive dephasing contractions also
protect all earlier prefixes.

Integer-directed arithmetic encloses the survival recurrence
$(1-2q)^L$, with endpoints reversed when converting survival to error.
For logarithms, range reduction writes the argument as $2^kr$ with
$1\le r<2$. Setting $z=(r-1)/(r+1)$ gives
\begin{equation}
 \log r=2\sum_{j=0}^{m-1}\frac{z^{2j+1}}{2j+1}+\mathcal R_m,
 \quad 0\le\mathcal R_m\le\frac{9z^{2m+1}}{4(2m+1)}.
 \label{eq:task-domain-log-bound}
\end{equation}
This follows by bounding the positive tail geometrically with
$z\le1/3$. Using $m=80$ and integer scale $10^{70}$, every product
and division is rounded outward; negative $k$ reverses the endpoints
of $k\log2$. These are interval bounds, not a floating-point grid test.

Checking all 1501 integer depths in the band gives achieved-error
margin above $3.61\times10^{-5}$, product cost at most
$204680\tau_c$, and encoded necessary cost above $383999\tau_c$.
The largest service duration is $4680\tau_c$. For the rectangle's
211 integer depths, the maximum achieved error is below $0.009959175$,
the minimum encoded necessary cost exceeds $317164\tau_c$, and
the maximum product cost and service duration are respectively
$204794.4\tau_c$ and $4794.4\tau_c$. Both domains therefore satisfy
the common budget and deadline strictly. They overlap at $L=4500$.
These domains are certified sufficient subsets, rather than optimized
boundaries of achievable tasks.

\section{Readout false positives and rare-event information}
\label{app:ab-noise-proof}
After ideal extraction and recovery, independently flip the delivered
fields $m_X,m_Z$ with probabilities $\epsilon_X,\epsilon_Z$.
The corrupted fields are never fed into recovery or service; the other
six fields retain their nominal laws. Define
\begin{equation}
 s_X=1-2\epsilon_X,\quad\rho_Z=1-2\epsilon_Z,
 \quad c_X=s_Xc,\quad h=s_X\rho_Zd,
 \label{eq:ab-flip-parameters}
\end{equation}
where $c,d$ are given in Eq.~\eqref{eq:product-cd}.
Convolving the four pre-flip probabilities gives
\begin{equation}
 P_\theta(m_X'=a,m_Z'=b)=\frac{1+ac_X-\theta abh}{4}.
 \label{eq:ab-flip-law}
\end{equation}
The remaining three fair fields and three deterministic fields contain
no label information, so this specifies the complete record divergence:
\begin{equation}
 K=\frac h2\left[\log\frac{1+c_X+h}{1+c_X-h}
                   +\log\frac{1-c_X+h}{1-c_X-h}\right].
 \label{eq:ab-exact-kl}
\end{equation}
Both KL directions are equal because changing the label swaps the two
values of $m_Z'$ at each $m_X'$. For fixed $0\le\epsilon_Z<1/2$
and $\epsilon_X=\kappa\lambda$, with finite $\kappa\ge0$, expansion
of $c=1-2\lambda+O(\lambda^2)$ and $d=\lambda+O(\lambda^2)$ gives
\begin{equation}
 K=\frac{\rho_Z}{2}
 \log\frac{2(1+\kappa)+\rho_Z}{2(1+\kappa)-\rho_Z}\,
 \lambda+O(\lambda^2).
 \label{eq:ab-scaled-floor-kl}
\end{equation}
The rare-sector mass is $[1+\kappa]\lambda+O(\lambda^2)$ and
retains a finite conditional contrast. At fixed
$0<\epsilon_X<1/2$, both $m_X'$ sectors instead have positive
limiting probability; their label differences are only $O(\lambda)$.
Expanding both logarithms in Eq.~\eqref{eq:ab-exact-kl} yields
\begin{equation}
 K=\frac{s_X^2\rho_Z^2}{2\epsilon_X(1-\epsilon_X)}
       \lambda^2+O(\lambda^3).
 \label{eq:ab-fixed-floor-kl}
\end{equation}
The squared Hellinger distance has the same respective positive orders,
as follows directly from Eq.~\eqref{eq:ab-flip-law}; thus the fixed-error
independent-probe costs have the corresponding inverse orders.
At either fully randomizing rate $1/2$, label dependence disappears.
The contrast-field error $\epsilon_Z<1/2$ reduces the signal, whereas
a fixed false-positive background in the rare-sector field $m_X$ changes its
asymptotic order. These delivered-field flips are label-independent
postprocessing of encoded records too, so they cannot increase the
encoded information cap. This is a record-noise model after ideal
correction, not a general noisy-extraction or state-preparation-and-measurement
(SPAM) threshold.

The exact finite working point illustrates why per-query information alone is
insufficient. At $\lambda=10^{-3}$, $\epsilon_X=0.01$, and $\epsilon_Z=0$,
Eq.~\eqref{eq:ab-exact-kl} gives equal directed divergences
$K=4.4076783\ldots\times10^{-5}$ nats per probe. Any plan that reserves the
minimum $4600\tau_c$ service cost under $B=250000\tau_c$ and charges
$r_{\rm product}=20$ fits at most 12270 probes, giving $nK<0.541$; the
displayed worst-service allowance $4784\tau_c$ leaves 12260. This is below the
necessary forward task information $1.48$ in
Eq.~\eqref{eq:b2-numeric-kl-bounds}. Therefore the ideal-record Fig.~2 task
guarantee does not extend to this charged noisy-record point.

A constructive noisy-record point follows without simulation. For
$Z=\boldsymbol1\{m_X'=-1\}m_Z'$, put
$q=(1-c_X)/2$ and $\mu=h/2$. An exponential Markov bound, with the
positive-model error enlarged to include a tie, gives both sign-test errors as
\begin{equation}
 \begin{aligned}
  \Pr_+\{T<0\},\ \Pr_-\{T\ge0\}&\le M^n,\\
  M&=1-q+\sqrt{q^2-\mu^2}.
 \end{aligned}
 \label{eq:ab-noisy-selector-bound}
\end{equation}
At the same $\lambda,\epsilon_X,\epsilon_Z$, odd/even Taylor bounds for the
exponentials and an exact rational squared-inequality check give
$M<0.999989104$. The bound
$\log(1-\delta)\le-\delta-\delta^2/2$, with
$\delta=1-0.999989104$, and an even Taylor upper bound on the remaining
exponential give $M^{300000}<0.038052$. Combining this error bound with the upper service-rate
endpoints in Eq.~\eqref{eq:full-box-service-rates} and
Eq.~\eqref{eq:robust-selector-risks} gives
$R_{{\rm sel},+}<0.009946$ and $R_{{\rm sel},-}<0.009522$, protecting every
prefix through $L=4600$. If $r_{\rm product}\le3$, the total cost is at most
$(300000\times3+4784)\tau_c=904784\tau_c$. Thus $B=10^6\tau_c$ admits this
product protocol while the unchanged encoded lower bound
$B_{\rm encoded}>1.14\times10^6\tau_c$ excludes encoded calibration. This is
a sufficient point, not a minimum sample count or a noisy-extraction/SPAM
certificate.

\end{document}

%% file: macros.tex
\newtheorem{theorem}{Theorem}
\newtheorem{lemma}{Lemma}

\newtheorem{proposition}{Proposition}